\documentclass[journal]{IEEEtran}
\usepackage{amssymb}
\usepackage{amsfonts}
\usepackage{epsf}
\usepackage{epsfig}
\usepackage{epstopdf}
\usepackage{array}
\usepackage{amsmath}
\usepackage{url}
\usepackage[algoruled, linesnumbered]{algorithm2e}
\usepackage{amsthm}
\usepackage{xifthen}
\usepackage{mathrsfs}
\usepackage[mathscr]{euscript}
\usepackage[sans]{dsfont}
\usepackage[dvipsnames]{xcolor}
\usepackage{graphicx}
\usepackage{caption}
\usepackage{subcaption}

\usepackage{epsf}
\usepackage{epsfig}

\usepackage{array}

\usepackage{times}
\usepackage{graphicx}
\usepackage{amsmath}
\usepackage{amssymb}
\usepackage{amsxtra}
\usepackage{here}
\usepackage{rawfonts}
\usepackage{times}
\usepackage{url}
\usepackage{cite}
\usepackage{amsthm}
\usepackage{graphicx}
\usepackage{caption}

\usepackage{mathtools}
\usepackage{url}
\usepackage[algoruled,linesnumbered]{algorithm2e}
\usepackage{glossaries}

\usepackage{cite} 

\usepackage[T1]{fontenc}
\usepackage[utf8]{inputenc}
\usepackage{authblk}

\newtheorem{proposition}{Proposition}

\newtheorem{remark}{Remark}

\theoremstyle{definition}
\newtheorem{definition}{Definition}

\ifCLASSINFOpdf
\else
\fi
\newcommand\mysymbol[3]{%
\protected\gdef#1{#2}%
\item[$#2$]#3}

\begin{document}
\title{Dynamic Connectivity Game for Adversarial Internet of Battlefield Things Systems}

\author{Nof~Abuzainab ~%\IEEEmembership{Member,~IEEE,}
        ~and~Walid~Saad,~\IEEEmembership{Senior~Member,~IEEE,}
         \vspace{-2.5ex}
\thanks{N. Abuzainab and W. Saad are with the department of Electrical and Computer Engineering, Virginia Tech, Blacksburg, VA, USA, e-mail: \{nof, walids\}@vt.edu}
\thanks{This research was sponsored by the Army Research Laboratory and was
accomplished under Grant Number W911NF-17-1-0021. The views and conclusions contained in this
document are those of the authors and should not be interpreted as representing the official policies, either
expressed or implied, of the Army Research Laboratory or the U.S. Government. The U.S. Government is
authorized to reproduce and distribute reprints for Government purposes notwithstanding any copyright
notation herein.}% <-this % stops a spac
\thanks{Copyright (c) 2012 IEEE. Personal use of this material is permitted. However, permission to use this material for any other purposes must be obtained from the IEEE by sending a request to pubs-permissions@ieee.org.}
}
\normalsize
\maketitle
\IEEEpeerreviewmaketitle
%%
%\begin{abstract}
%This paragraph shall summarize the contents of the paper
%in short terms.
%\end{abstract}

\begin{abstract}
In this paper, the problem of network connectivity is studied for an adversarial Internet of Battlefield Things (IoBT) system in which an attacker aims at disrupting the connectivity of the network by choosing to compromise one of the IoBT nodes at each time epoch.  To counter such attacks, an IoBT defender attempts to reestablish the IoBT connectivity by either deploying new IoBT nodes or by changing the roles of existing nodes.  This problem is formulated as a dynamic multistage Stackelberg connectivity game that extends classical connectivity games and that explicitly takes into account the characteristics and requirements of the IoBT network. In particular, the defender's payoff captures the IoBT latency as well as the sum of weights of disconnected nodes at each stage of the game. Due to the dependence of the attacker's and defender's actions at each stage of the game on the network state, the feedback Stackelberg solution (FSE) is used to solve the IoBT connectivity game.  Then, sufficient conditions under which the IoBT system will remain connected, when the FSE solution is used, are determined analytically. Numerical results show that the expected number of disconnected sensors, when the FSE solution is used, decreases up to $46\%$ compared to a baseline scenario in which a Stackelberg game with no feedback is used,
 and up to $43\%$ compared to a baseline equal probability policy. \vspace{-0.3 cm}
\end{abstract}

%The attacker's payoff at each stage is the sum of weights of disconnected nodes minus the cost of compromising a node.The defender's payoff is expressed in terms of the utility of deploying a new node minus the sum of its costs which are the sum of weights of disconnected nodes, the time required to deliver the information to the GS and the cost of deploying a new node. 

%In the studied IoBT connectivity problem, the defender's actions are coupled to the attacker's actions at each stage of the game. Thus, the problem is formulated in particular as a multitage Stackelberg game where, at each stage of the game, the attacker acts as a leader, and the defender acts as a follower.%\vspace{0.2 cm}

\section{Introduction}

\IEEEPARstart{T}{he} Internet of Things (IoT) is expected to revolutionize the military battlefield in various aspects \cite{IoBT, IoBTone, IoBTtwo}. By interconnecting all military units, including soldiers and vehicles, with various IoT devices, sensors, and actuators, the IoT provides autonomy in the battlefield and increases the efficiency of military networks.
An IoT-enabled battlefield will allow military commanders to acquire instanteneous information on the status of the military units. For instance, wearables can provide instant updates on the situation of soldiers, and sensors mounted on vehicles can provide real-time information on the status of each vehicle.
Another important IoT feature that makes it suitable for the battlefield is its support for mobile crowdsensing. In mobile crowdsensing,  various IoT devices such as handheld devices, wearables, vehicles, and sensors collaborate in sensing a particular type of information. In traditional military networks, on the other hand, dedicated sensors are deployed for each application.  Thus, a dense deployment of IoT devices can provide more accurate and detailed information about the battlefield, which can, in turn, allow building comprehensive situation awareness and enabling more accurate decision making. This imminent integration of the IoT with military networks forms the nexus of the so-called \emph{ Internet of Battlefield Things (IoBT)}\cite{IoBT}.

Naturally, in an IoBT, connectivity is very critical for the successful operation of the military network as it is essential to maintain the autonomy of the system. Military missions, such as  surveillance and situational awareness, will heavily rely on the information collected for the battlefield, and, thus, any disconnection in the IoBT system will result in inaccurate decision making and poor situational awareness. In fact, the IoBT is more vulnerable than commercial IoT networks due to the \emph{adversarial nature of the battlefield}, in which the devices are continuously subject to security attacks. Moreover, IoBT devices are typically small and low-cost devices that do not support strong security mechanisms, and, hence, they can be easily compromised by adversaries. The vulnerability of the IoBT devices necessitates the design of novel security solutions that are robust to adversaries and that can maintain the connectivity of the IoBT in adversarial settings.  

Connectivity reconstruction solutions were initially designed for wireless networks such as in \cite{wirelessconnectivity} in which the nodes select their transmission powers to maintain network connectivity. In \cite{cloudconnectivity} and \cite{starconnectivity}, connectivity establishment mechanisms are proposed to reestablish connectivity between sensors that were isolated, due to faults or attacks, and a central sink in a sensor network. In \cite{cloudconnectivity}, the connectivity problem is formulated as a single leader, multiple followers Stackelberg game in which a cloud acts as the leader and chooses to activate sleep nodes in order to maintain full connectivity, whereas the sleep nodes act as followers with each seeking to maximize the number of isolated nodes that it reconnects to the network. In \cite{starconnectivity}, stochastic geometry is used to design a relay-based connectivity recovery scheme for a wireless sensor network whose the goal is to optimize the tradeoff between the number of selected relays and the energy spent to restore connectivity.
In \cite{industrial}, the authors derive conditions for regional connectivity in an IoT industrial system while optimizing sensor coverage. The work In \cite{IoTjournal4} proposes a dynamic clustering and routing algorithm to maintain connectivity and achieve energy efficiency in a large scale sensor network.
In \cite{IoTjournal2}, a dynamic mobile-aware IoT topology control scheme, based on a potential game, is proposed in order to optimize IoT connectivity.
The work in \cite{IoTjournal1} introduces a resilience mechanism to maintain percolation-based connectivity in an IoT network in which an adversary seeks to attack highly connected IoT nodes in order to achieve the maximum possible damage. In the model of \cite{IoTjournal1},  the IoT nodes report a one bit estimate of their attack status to a common fusion center. Then, the objective of the fusion center is to choose the nodes to survey such that the number of nodes with highest degree under attack is kept below a required threshold. The problem is formulated as a zero-sum game between the fusion center and the attacker.
In \cite{IoTjournal3}, a time-reversal scheme is proposed in an IoT network to enable connectivity between devices with heterogeneous bandwidth requirements.

However, most of these existing works \cite{wirelessconnectivity,cloudconnectivity,starconnectivity,industrial, IoTjournal4, IoTjournal2, IoTjournal1, IoTjournal3} consider the connectivity problem in conventional sensor networks in which all the nodes are simple sensors of the same type and capabilities, whereas in the IoBT, the nodes can have heterogeneous roles and capabilities. In fact, each IoBT device can possess multiple sensors each of which is collecting different types of information. Thus, the importance of each device is dependent on the number of types of information it is sensing.  Further, the IoBT will integrate high end nodes, commonly known as sinks, that collect the different information from the IoBT devices and perform complex operations in order to obtain useful information needed by the military commanders \cite{IoBT, IoBTone, IoBTtwo}. Thus, the effect of disconnection on the IoBT depends on the type of the node that gets isolated from the network. Further, prior art such as in \cite{wirelessconnectivity,cloudconnectivity,starconnectivity,industrial, IoTjournal4, IoTjournal2, IoTjournal1, IoTjournal3} does not adequately capture the dynamics of interaction between defenders and adversaries in a battlefield. Thus, there is a need to introduce new dynamic connectivity solutions that consider the heterogeneity of the IoBT nodes and dynamically adapt to the actions of adversaries in the battlefield.

The main contributions of this paper are summarized next:
\begin{itemize}
 \item We develop a novel adaptive framework for dynamically optimizing the connectivity of an adversarial IoBT network. 
In particular, we consider the connectivity problem in an IoBT that includes a set of heterogeneous devices that sense different types of information. The IoBT devices must transmit their information, through intermediary local sinks, to the general sink. We consider an adversarial IoBT in which an attacker is interested in causing disconnection to the network by choosing to compromise one of the IoBT nodes at each time epoch.  Meanwhile, the IoBT operator acts as a defender that strives to maintain the connectivity of the IoBT network by either deploying new IoBT nodes or changing the roles of the nodes.
The objective of the attacker and the defender is to maximize their sum of payoffs until the end of the military operation.
%The objective of each of the attacker and the defender is to maximize its finite horizon sum of payoffs.
%The attacker's payoff at each stage is the sum of weights of disconnected nodes minus the cost of compromising a node.The defender's payoff is expressed in terms of the utility of deploying a new node minus the sum of its costs which are the sum of weights of disconnected nodes, the time required to deliver the information to the GS and the cost of deploying a new node. The defender also must maintain the number of IoBT devices sensing the same type of information above a certain required threshold. 
\item We formulate the connectivity problem in the IoBT using the framework of  \emph{connectivity games}  \cite{Connectivity} which are game-theoretic frameworks suitable for addressing problems that involve the maintenance and restoration of a network in  presence of adversaries. 
However, in classical connectivity games, the sole objective is to restore or maintain the network connectivity, whereas in the IoBT, there are other performance metrics that must be considered such as the latency of communication. Thus, we propose a novel IoBT connectivity game that is tailored to the characteristics and requirements of the IoBT.
% In particular, In the proposed game, the defender's utility captures the IoBT network latency as well as the sum of weights
In particular, the attacker's payoff is expressed as the sum of weights of disconnected nodes minus the cost of compromising a node. The defender's payoff, on the other hand, is expressed as the utility of deploying a new node minus  the sum of weights of disconnected nodes, the time required to deliver the information to the IoBT general sink, and the cost of deploying a new node.
Further, in the studied IoBT connectivity problem, the defender must maintain the number of IoBT devices sensing the same type of information above a certain required threshold. Thus, the defender's strategy set is coupled with the attacker's action at each time epoch. Consequently, we cast the problem as a dynamic multistage Stackelberg connectivity game in which, at each stage of the game, the attacker acts as a leader, and the defender acts as a follower. Due to the dependence of the attacker's and defender's actions in each stage of the game on the network state, the feedback Stackelberg equilibrium (FSE) is used to solve the IoBT connectivity game. 
%Further, we modify the linear programming solution that is typically used to find the mixed strategy Stackelberg equilibrium in order to take into account the defender's coupled actions to the attacker's actions.
\item We analytically derive sufficient conditions for the IoBT network to remain connected at each stage of the game when the FSE solution is used.  Numerical results show that the expected number of disconnected sensors, when the FSE solution is used, decreases up to $43\%$ compared to a baseline scenario in which a Stackelberg game with no feedback is used,
 and up to $46\%$ compared to a baseline equal probability policy. 
\end{itemize}

The paper is organized as follows: Section I describes the adversarial IoBT system model.  Section II presents the formulation of the IoBT connectivity game. Section III presents the feedback Stackelberg solution of the IoBT game. 
Section IV presents sufficient connectivty conditions of the IoBT network when the FSE is used. Section V presents the simulations results and analysis. Finally, conclusions are drawn in Section VI. A complete list of the notations used is in Appendix A.

%and the attacker
%We formulate the connectivity problem in the IoBT using the framework of "connectivity games". Connectivity games are game theoretic models for situations that involves the maintaince and restoration of network in the presence of adversaries. Connectivity games are dynamic games in
%However, in classical connectivity games, the sole objecive is to restore or maintain the network connectivty, whereas in IoBT, there are other performance metrics that must be considered for IoBT especially latency. Thus, we propose an IoBT connectvity game that extends classical connectivity games and that takes into account the characteristics and requirements of the IoBT. We formulate the IoBT connectivity game as a discrete time dynamic Stackelberg game with a finite number of stages. In our problem, the players are the defender and the attacker who choose to compromise one of the node Then, we propose a solution approach based on the feedback Stackelberg solution concept as it adequately captures
\vspace{-0.3 cm}

\section{System Model}

%\printglossary[title={List of Symbols}]

Consider an IoBT network composed of a set $\mathcal{D}$ of heterogeneous devices that can be of different types within a set $\mathcal{K}$ of size $K$. Each IoBT device can possibly represent a vehicle, a drone, a robot, a surveillance camera,  a sensor dedicated for a certain type of application, a sensor-actuator pair or soldier equiped with wearable sensors. Each device of type $\tau \in \mathcal{K}$ encompasses $N_\tau$ sensors (and their corresponding actuators) sensing a subset $\mathcal{H}_\tau$ of a set $ \mathcal{I}$ of types of information. Due to the heterogeneity of the IoBT nodes, in terms of roles and capabilities, we consider a hierarchical tree structure \cite{msn1} \footnote{Although the hierarchical network structure is chosen, in general, our proposed approach can accomodate any network topology}. The hierarchical IoBT structure provides scalability and allows the system operator to easily add new devices, which is suitable for a large-scale IoBT system. The area that the IoBT network spans is divided into subareas $A_1$, $A_2$,...$A_H$. Within each area $A_h$, devices sensing the same type of information $j \in  \mathcal{I}$ are organized into a cluster $\mathcal{D}_{jh}$. Thus, an IoBT device equipped with multiple sensors can belong to several clusters. Within each cluster, one of the devices is chosen to be a cluster head (CH), and, thus, the rest of the devices transmit their sensed data  to the CH. The CH then collects the information received from the devices in the cluster and sends it to a local sink (LS) serving subarea $A_h$. 

In each subarea $A_h$, multiple LSs can be deployed for redundancy. At any time epoch $t$, only one LS  is activated in each subrea $A_h$. Deploying redundant LSs ensures that there is a substitute for the activated LS in case of failure or malfunction. At each time epoch $t$, each activated LS processes its information and performs more sophisticated operations such as augmented sensing and extraction of useful information as requested by the global sink (GS). The GS is a high end node that eventually processes the information received from the activated LSs in order to identify events requested by the military commanders and provide situational awareness. Since the nodes in the considered IoBT are of heterogeneous capabilities and roles,  each node $i$ is assigned a weight $w_i$ depending on its importance. The weight $w_i$ of each device $i \in \mathcal{D}$ of type $\tau$ is measured in terms of the number of different sensors that the device includes i.e. $w_i=N_\tau$. LSs, on the other hand, perform more sophisticated operations. Thus,  each LS $i$  is assigned a weight $w_{L,i}$ that is higher than the weights of the devices i.e.  $w_{L,i} > \max_{1 \leq \tau \leq K} N_\tau$.
\begin{figure}[t]
	\centering
	\includegraphics[width=8 cm,height=5 cm,angle=0]{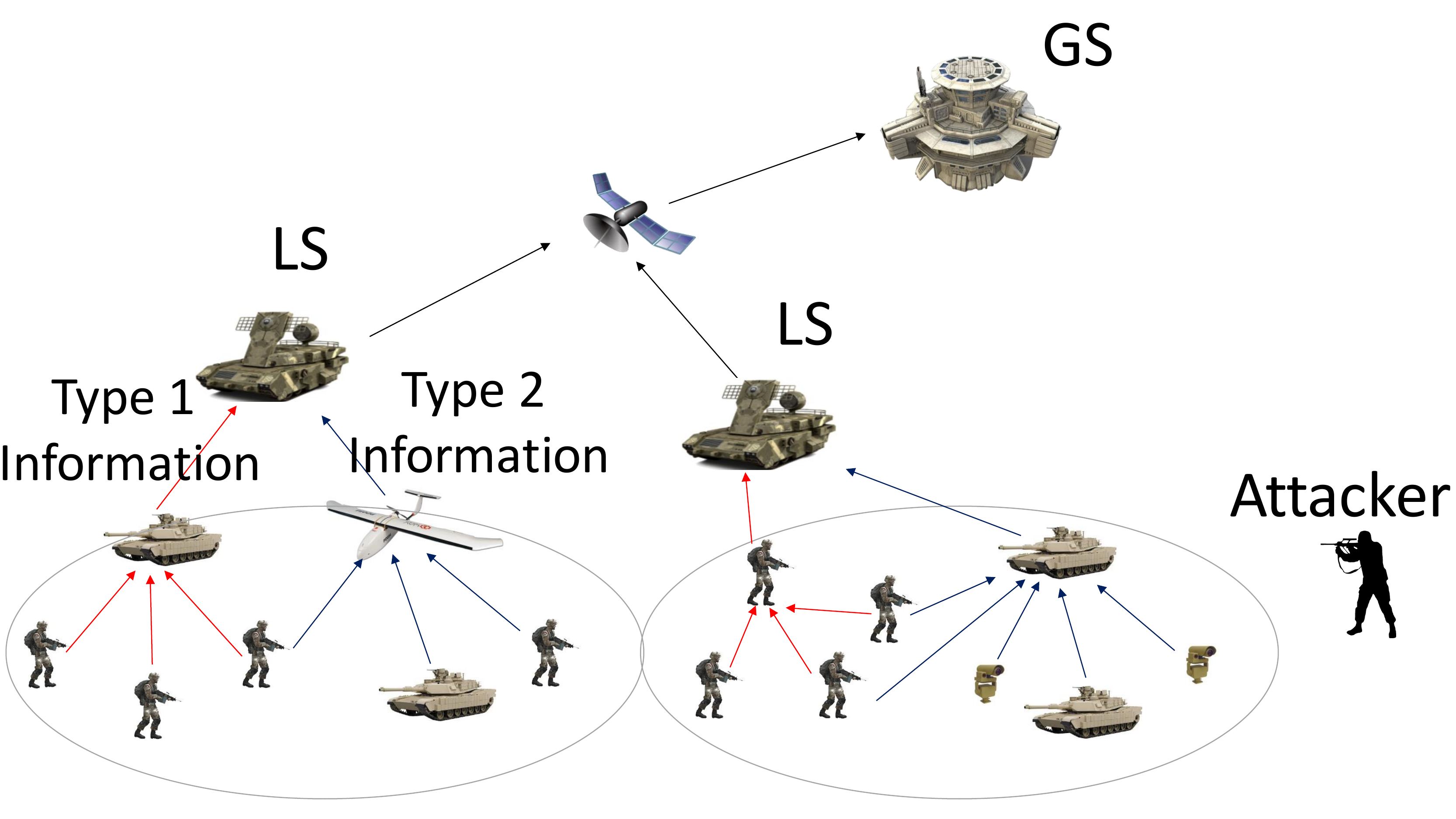}
	\caption{An example of the considered IoBT hierarchical system with two subareas and two types of information.
	}\vspace{-0.4 cm}\label{sysmodel}
\vspace{-0.4 cm}
\end{figure}

In this IoBT, an attacker is interested in minimizing the connectivity of the network to prevent the GS from detecting important events thus ultimately impairing its decisions.
To achieve this goal, the attacker chooses at each time $t$ to compromise i.e. gain control over one of the nodes in $\mathcal{B}=\cup_{h=1}^H \mathcal{L}_h \cup \mathcal{D}$ where  $\mathcal{L}_h$ is the set of LSs in subarea $A_h$.
In order to compromise each node, an attacker needs to spend time and computations to complete the attack \cite{Compromise}.  Thus, it is assumed at any time $t$, the attacker can only compromise one of the IoBT nodes.
Thus, at time $t$, the attacker chooses the node which maximizes its payoff which is expressed as the sum of weights of all nodes that will be disconnected from the GS, and the cost of compromising each node. This cost pertains to the resources needed to compromise any given targeted node.  Let $c_{\tau}$ be the cost of compromising device of type $\tau$ and $c_L$ be the cost of compromising an LS.  The attacker also incurs additional costs $c_{CH}$ and $c_{aL}$ in order to determine the CH of each cluster or the activated LS in each subarea. The costs $c_{CH}$ and $c_{aL}$ can represent, for example, the security costs of intercepting the beacon messages sent by the CH or the activated LS to the remaining devices.
Thus, the total cost of attacking device $i$ of type $\tau$ in subarea $A_h$ is given by: $c_i=c_\tau+\sum_{j=1}^Mx_{ijh}c_{CH}$ where{$M=|\mathcal{I}|$} and $x_{ijh}=1$ indicates that device $i$ is the CH of cluster $\mathcal{D}_{jh}$ or $x_{ijh}=0$, otherwise. The total cost incurred by attacking LS $i$ in subarea $A_h$ is given by: $c_{L,i}=c_{LS}+y_{ih}c_{aL}$ where $y_{ih}$ is the indicator that LS $i$ is activated in subarea $A_h$. In order to thwart the attacks made at each time epoch $t$, the defender can choose one of the following actions:
%Thus, for example, the cost of destroying a local sink will be higher than the cost of destroying a sensor.
\begin{enumerate}
 \item[1)]  Deploys a new device of type $\tau$ in subarea $A_h$;
\item[2)]  Changes the cluster head in cluster $\mathcal{D}_{jh}$;
\item[3)]  Changes the activated LS in subarea $A_h$;
\item[4)]  Deploys a new LS in area $A_h$.
\end{enumerate}
Action $1)$ helps in maintaining the number of sensors necessary to maximize the amount of useful information gathered within an area. Actions $2)$ and $3)$ ensure the robustness of the network in case the currently activated LS or CH fails or is destroyed by the attacker. In practice, the newly deployed devices are typically brought from a warehouse that is in the proximity of the battlefield.
Action $4)$ ensures that there will always exist an LS that could serve the sensors in any subarea in case the activated LS fails or is compromised by the attacker. For actions 1)  and 4), the defender will incur a cost of deploying a device or an LS. Let $d_{\tau}$ be the cost of deploying a device of type $\tau$ and let $d_L$ be the cost of deploying an LS.
In an IoBT, the newly deployed devices and LSs are intially stored, prior to deployment, in a storage facility (or a military base) that is in the proximity of the IoBT network that is assumed to be secured from the attacker.

%In order to obtain the necessary information for application $a_i$ in a certain subarea $A_h$, the defender must maintain the expected number of sensors  $\mathbb{E}[N_{il,t}]$ at time epoch $t$  above a certain threshold $N_{l,\text{th}}$. 

The objective of the defender is to maximize a payoff that captures the difference between the achieved utility and the sum of its costs until the end of the military operation at time epoch $T$. This maximization will be subject to the constraint that the  number of sensors  $N_{jh}(t)$ sensing information of type $j$ in subarea $A_h$ at each time epoch $t$  does not fall below a certain threshold $N_{\text{th},jh}$.  This constraint ensures that the GS as well as the LSs obtain the necessary information of type $j$ in a certain subarea $A_h$. The utility achieved from deploying a device of type $\tau$ in subarea $A_h$ is expressed in terms of the number of clusters that will restore their number of sensors above the threshold and is given by $u_\tau=\sum_{j=1}^MI(j \in \mathcal{H}_\tau)I(N_{jh}(t)<N_{\text{th},jh})$ where $I(.)$ is an indicator function. The utility achieved from deploying an LS in subarea $A_h$ is given by $u_L=B-L_h$ where $L_h=|\mathcal{L}_h|$ and  $B$ is a constant that reflects the recommended number of LSs in each subarea. The defender's utility is the utility of deploying a new device or an LS. 
The defender's cost at each time epoch $t$ is expressed in terms of the sum of weights of disconnected nodes, the time spent to deliver the information to the GS, and the cost of deploying a new node.

%Due to the dependence of the payoffs  on the actions of both the attacker and the defender, the problem of finding the optimal actions of both the attacker and the defender will be formulated as a \emph{noncooperative game} \cite{gametheory} where the players are the attacker and the defender. Further,  the actions of both players  affect the shape of the IoBT network since they involve the addition and deletion of nodes as well as changing of roles of nodes (such as changing the CHs and LSs). Thus, the game will be formulated in particular as a positional/connectivity game \cite{Connectivity} as follows.
Due to the clear dependence between the goals and the actions of the attacker and the defender as well as the impact of the attacker and defender's actions on the IoBT network graph, the problem will be formulated as a \emph{noncooperative positional game} \cite{Connectivity} and  \cite{gametheory}, as explained next.
\vspace{-0.1 cm}
\section{IoBT Connectivity Game}

Connectivity games are game-theoretic models \cite{Connectivity} that capture situations which require the maintainenance and restoration of the normal operations of a given network. Connectivity games typically involve two players: a constructor who is responsible for restoration of nodes as well as the addition of new nodes, and a destructor who removes nodes from the network. The constructor in our game is the IoBT defender whereas the destructor is the attacker.
A connectivity game \cite{Connectivity} is an interactive game in which the constructor and the destructor play in alternation until one of the players wins the game. The winning condition for the constructor involves  maintaining the connectivity of the network. In particular, there are two types of objectives considered in classsical connectivity games \cite{Connectivity}: 1) A safety objective in which the constructor must maintain the connectivity of the network in every step of the game and 2) A reachability objective in which the constructor must obtain a connected network starting from a disconnected network.

However, in the IoBT setting, the objective is not only to maintain the network connectivity but also to maximize the network efficiency (for example in terms of energy efficiency and latency). Further, in the IoBT network, devices sense different types of information, and in order to obtain the necessary information of each type in a certain area, there is a need to ensure that the 
number of devices sensing the same type of information does not drop below a required threshold. Moreover, in a real-world IoBT, there is a cost incurred when a device is destroyed by the attacker or deployed by the defender, which is not considered in a classical connectivity game \cite{Connectivity}. The heterogeneity of the IoBT devices, in terms of their importance and roles, is also not taken into account in classical connectivity games \cite{Connectivity}.

Given these requirements and characteristics of the IoBT network, we consider an IoBT connectivity game that extends classical connectivity games.
The IoBT connectivity game is formulated as a discrete-time deterministic dynamic game ($\mathcal{P},\mathcal{T}, \mathcal{X}, (\mathcal{S}_{a,t},\mathcal{S}_{d,t})_{t \in \mathcal{T}})$ with  a finite number of stages, where the set of players $\mathcal{P}$ includes the attacker and the defender, and the set of stages $\mathcal{T}=\{1,2,...,T\}$.
In this IoBT connectivity game, the defender must observe the attacker's action before choosing its optimal action in order to maintain the number of devices in each area above the required threshold.  Thus, the IoBT connectivity game is formulated as a Stackelberg game in which, at each stage $t$ of the game, the attacker acts as the leader, and the defender acts as the follower. The state space $\mathcal{X}=\mathcal{X}_a \times \mathcal{X}_d$ is the set of all IoBT networks observed by the attacker and the defender up to stage $T$. The state of the game at stage $t$ is $\psi_t=(\psi_{a,t},\psi_{d,t}) \in \mathcal{X}$ where $\psi_{a,t}$ is the network observed by the attacker and $\psi_{d,t}$ is the network observed by the defender. The network state observed by the attacker is given by: $\psi_{a,t}=(\mathcal{D}_{a}(t),\{\mathcal{L}_{a,h}(t), 1 \leq h \leq H\}, \{\mathcal{D}_{a,jh}(t), 1 \leq j \leq I, 1 \leq h \leq H\}, \{f_{a,jh}(t),  1 \leq j \leq I, 1 \leq h \leq H\},  \{s_{a,h}(t), 1 \leq h \leq H\} )$ where $\mathcal{D}_{a}(t)$  represents the set of devices, $\mathcal{L}_{a,h}(t)$ is the set of LSs in subarea $A_h$, $\mathcal{D}_{a,jh}(t)$ is the cluster of devices sensing information type $j$ in subarea $A_h$, $f_{a,jh}(t)$ is the index of the device that is the CH of $\mathcal{D}_{a,jh}(t)$, and $s_{a,h}(t)$ is the index of the activated LS in subarea $A_h$. Similarly, the network state $\psi_{d,t}$ observed by the defender is given by: $\psi_{d,t}=(\mathcal{D}_{d}(t),\{\mathcal{L}_{d,h}(t), 1 \leq h \leq H\}, \{\mathcal{D}_{d,jh}(t), 1 \leq j \leq I, 1 \leq h \leq H\}, \{f_{d,jh}(t),  1 \leq j \leq I, 1 \leq h \leq H\},  \{s_{d,h}(t), 1 \leq h \leq H\} )$ where $\mathcal{D}_{d}(t)$  represents the set of devices, $\mathcal{L}_{d,h}(t)$ is the set of LSs in subarea $A_h$, $\mathcal{D}_{d,jh}(t)$ is the cluster of devices sensing information type $j$ in subarea $A_h$, $f_{d,jh}(t)$ is the index of the device that is the CH of $\mathcal{D}_{d,jh}(t)$, and $s_{d,h}(t)$ is the index of the activated LS in subarea $A_h$. $f_{a,jh}(t)$ and $f_{d,jh}(t)$ are set to be zero if the CHs of $\mathcal{D}_{a,jh}(t)$ and $\mathcal{D}_{d,jh}(t)$ are compromised, respectively. Similarly, $s_{a,h}(t)$ and $s_{d,h}(t)$ are set to be zero if the activated LS of subarea $A_h$ is compromised in $\psi_{a,t}$ and $\psi_{d,t}$ respectively. The defender and the attacker are assumed to have perfect knowledge on the IoBT network. Assuming perfect knowledge by the attacker about the IoBT network allows the defender to account for the worst-case scenario as is typical in existing works such as in \cite{IoTjournal1}.
% It is assumed that the attacker has imperfect knowledge of the IoBT network. That is, at each time epoch $t$, it can determine perfectly the set of devices in each cluster $\mathcal{D}_{jh}$ and the set of LSs in each area $A_h$. However, the attacker can not determine pefectly the CH in each cluster $\mathcal{D}_{jh}$ and which LS is activated in area $A_h$. Thus, the attacker maintains a belief $\boldsymbol{p}_{jh}$ about the CH in each cluster $\mathcal{D}_{jh}$ and a belief $\boldsymbol{r}_h$ about the activated LS in each subarea $A_h$

In our game, the set of pure strategies of the attacker $\mathcal{S}_{a,t}(\psi_t)$ at stage $t$ is $\mathcal{S}_{a,t}(\psi_t)=\{a_{d,i}, i \in \mathcal{D}_{a}(t)\} \cup \{a_{L,lh}, l \in \mathcal{L}_{a,h}(t),  1 \leq h \leq H\}$, where action $a_{d,i}$ corresponds to destroying device $i$, and action $a_{L,lh}$ corresponds to destroying LS $l$ in subarea $A_h$. Due to the constraint on the number of devices in each cluster, the strategy set of the defender at each stage $t$ is coupled to the attacker's action $a_t$ and is a function of the network state $\psi_t$. Hence, the strategy set $\mathcal{S}_d(\psi_t,a_t)$ of the defender is
%$\mathcal{S}_d(\psi_t,a_t)=\{b_{d,il}\}$
%$\mathcal{S}_d(\psi_t,a_t)=\{b_{CH}, b_{LS}\}\cup\{b_{d,il}, 1 \leq i \leq K, 1\leq l \leq L\} \cup \{b_{L,l}, 1\leq l \leq L\}$
% \[\smaller
%    \mathcal{S}_d(\psi_t,a_t)=\left\{
%                \begin{array}{ll}
%                \hspace{-0.2 cm} \{b_{d,il}\} \hspace{0.4 cm}  \textrm{if} \hspace{0.2 cm} a_t=a_{d,j},  j \in \mathcal{D}_l, \tau(j)=\tau_i,\\ \hspace{1.6 cm}\exists k \hspace{0.1 cm} s.t. \hspace{0.1 cm} I_k \in \mathcal{H}_i, N_{kl}(t)=N_{th,kl}\\
%                \hspace{-0.2 cm}\{b_{CH}, b_{LS}\}\cup\{b_{d,il}, 1 \leq i \leq K, 1\leq l \leq L\}\\ \cup \{b_{L,l}, 1\leq l \leq L\}\hspace{1 cm} \textrm{otherwise}\\
%                 
%                \end{array}
%              \right.
%  \]
\vspace{-0.1 cm}

\small
\begin{equation}
    \mathcal{S}_{d,t}(\psi_t,a_t)=\left\{
                \begin{array}{ll}
                 \{b_{d,\tau h},  \forall \tau \hspace{0.1 cm} \mid \hspace{0.1 cm} j \in \mathcal{H}_\tau \forall j \in \mathcal{Y}_{ih}(t)\} \hspace{0.2 cm}  \textrm{if} \hspace{0.2 cm} E_1(a_t),\\
               \mathcal{Q}_d, \hspace{1 cm} \textrm{otherwise,}\\
                 
                \end{array}
              \right. \label{defset}
  \end{equation}
%\begin{equation}\vspace{-0.3 cm}\label{defset}\end{equation}
\normalsize
%
%\begin{itemize}
%\item  \textrm{if} $\hspace{0.2 cm} a_t=a_{d,j}, \textrm{and} \hspace{0.1 cm} \exists k \hspace{0.1 cm} s.t. \hspace{0.1 cm}j \in \mathcal{D}_{kl}, N_{kl}(t)=N_{th,kl}$\\
%$ \mathcal{S}_d(\psi_t,a_t)= \{b_{d,il},  \forall i s.t I_k \in \mathcal{H}_i \forall k \in \mathcal{Y}_t\}$
%\end{itemize}
where condition $E_1(a_t)$ is $a_t=a_{d,i}$, \hspace{0.1 cm} $\exists j \hspace{0.1 cm} \textrm{s.t.} \hspace{0.1 cm}i \in \mathcal{D}_{a,jh}(t), N_{a,jh}(t) \leq N_{\textrm{th},jh}$,
 the set 
$\mathcal{Y}_{ih}(t)=\{j \in \mathcal{I} \hspace{0.1 cm} \mid \hspace{0.1 cm}i \in \mathcal{D}_{a,jh}(t), 
N_{d,jh}(t)<N_{\textrm{th},jh}\}$, $N_{d,jh}(t)$ is the number of devices in cluster $\mathcal{D}_{d,jh}(t)$ in network $\psi_{d,t}$, $N_{a,jh}(t)$ is the number of devices in cluster $\mathcal{D}_{a,jh}(t)$ in network $\psi_{a,t}$,  
and the set $\mathcal{Q}_d$ is the set of all possible strategies of the defender given by
%$\mathcal{Q}_d=\{b_{c,ijh}, i \in \mathcal{D}_{jh}, 1 \leq j \leq M, 1 \leq h \leq H\} \cup\{b_{d,\tau h}, 1 \leq \tau \leq K, 1\leq h \leq H\}\nonumber\\
%\cup  \{b_{a,lh}, l \in \mathcal{L}_h, 1\leq h \leq H\}  \cup  \{b_{L,h}, 1\leq h \leq H\}, $
%\end{eqnarray}
\begin{eqnarray}
&&\hspace{-0.7 cm}\mathcal{Q}_d=\{b_{c,ijh}, i \in \mathcal{D}_{d,jh}(t), 1 \leq j \leq M, 1 \leq h \leq H\} \nonumber\\
&&\hspace{0.1 cm}\cup\{b_{d,\tau h}, 1 \leq \tau \leq K, 1\leq h \leq H\} \cup  \{b_{L,h}, 1\leq h \leq H\}\nonumber\\
 &&\hspace{0.1 cm}\cup  \{b_{a,lh}, l \in \mathcal{L}_{d,h}(t), 1\leq h \leq H\}. \label{kd}
\end{eqnarray}

Action $b_{d,\tau h}$ corresponds to deploying a new device of type $\tau$ in subarea $A_h$, action $b_{c,ijh}$ corresponds to assigning device $i$ to be the CH of cluster $\mathcal{D}_{d,jh}(t)$, action $b_{L,h}$ corresponds to deploying a new LS in subarea $A_h$, and action $b_{a,lh}$ corresponds to activating LS $l$ in subarea $A_h$. According to (\ref{defset}), if the attacker destroys a device and causes the number of devices in some clusters to drop below the required threshold, the strategy set of the defender will only include the actions of deploying a device that restores the number of devices in each affected cluster to the required threshold. Otherwise, the defender can choose to either change the CHs, change the LSs, deploy a new device, or deploy a new LS. 
%\mathcal{S}_{d,t}(\psi_t,a_t)
The evolution of the attacker's state $\psi_{a,t}$ is given by

\small
\begin{equation}
 \hspace{-1 cm} \mathcal{D}_{a}(t+1)  =\left\{
                \begin{array}{ll}
                \mathcal{D}_{a}(t)\setminus\{i\} \hspace{0.2 cm}  \textrm{if} \hspace{0.2 cm} a_t=a_{d,i} \hspace{0.1 cm} b_t \neq b_{d,\tau h}  ,\\
                \mathcal{D}_{a}(t), \hspace{1 cm} \textrm{otherwise,}\\
                \end{array}
              \right. \label{defset}
  \end{equation}

\begin{equation}
  \mathcal{L}_{a,h}(t+1)  =\left\{
                \begin{array}{ll}
                \mathcal{L}_{a,h}(t)\setminus\{l\} \hspace{0.2 cm}  \textrm{if} \hspace{0.2 cm} a_t=a_{L,lh}, \hspace{0.1 cm} b_t \neq b_{L, h}, \\
                \mathcal{L}_{a,h}(t), \hspace{1 cm} \textrm{otherwise,}\\
                \end{array}
              \right. \label{defset}
  \end{equation}

%\begin{equation}
%  \mathcal{D}_{a,jh}(t+1)  =\left\{
%                \begin{array}{ll}
%                \mathcal{D}_{a,jh}(t)\setminus\{i\} \hspace{0.2 cm}  \textrm{if} \hspace{0.2 cm} a_t=a_{d,i} \hspace{0.1 cm} i \in \mathcal{D}_{a,j'h'}(t)\\
%                \hspace{2.6 cm} b_t \neq b_{d,\tau h'} \hspace{0.1 cm} j = j'  \hspace{0.1 cm}  \text{and}  \hspace{0.1 cm} h = h'  ,\\
%                \mathcal{D}_{a,jh}(t), \hspace{1 cm} \textrm{otherwise,}\\
%                \end{array}
%              \right. \label{defset}
%  \end{equation}

\begin{equation}
  \mathcal{D}_{a,jh}(t+1)  =\left\{
                \begin{array}{ll}
                \mathcal{D}_{a,jh}(t)\setminus\{i\} \hspace{0.2 cm}  \textrm{if} \hspace{0.2 cm} a_t=a_{d,i} \hspace{0.1 cm} i \in \mathcal{D}_{a,jh}(t)\\
                \hspace{2.4 cm} b_t \neq b_{d,\tau h},\\
                \mathcal{D}_{a,jh}(t), \hspace{1 cm} \textrm{otherwise,}\\
                \end{array}
              \right. \label{defset}
  \end{equation}

%\begin{equation}
%  f_{a,jh}(t+1)  =\left\{
%                \begin{array}{ll}
%                0 \hspace{1 cm}  \textrm{if} \hspace{0.2 cm} a_t=a_{d,i}, \hspace{0.1 cm}  i \in \mathcal{D}_{a,jh}(t), f_{a,jh}(t)=i,\\
%                \hspace{1.6 cm}   b_t \neq b_{c,i'jh} \hspace{0.1 cm}, j=j', \hspace{0.1 cm}\text{and} \hspace{0.1 cm} h=h'  ,\\\\
%                 i \hspace{1 cm} \textrm{if}\hspace{0.1 cm} b_t= b_{c,ijh}, a_t \in \mathcal{S}_a(\psi_t) \\\\
%                f_{a,jh}(t), \hspace{0.1 cm} \textrm{otherwise,}\\
%                \end{array}
%              \right. \label{defset}
%  \end{equation}

\begin{equation}
  f_{a,jh}(t+1)  =\left\{
                \begin{array}{ll}
                0 \hspace{1 cm}  \textrm{if} \hspace{0.2 cm} a_t=a_{d,i}, \hspace{0.1 cm}  i \in \mathcal{D}_{a,jh}(t), f_{a,jh}(t)=i,\\
                \hspace{1.6 cm}   b_t \neq b_{c,i'jh} \\\\
                 i \hspace{1 cm} \textrm{if}\hspace{0.1 cm} b_t= b_{c,ijh}, a_t \in \mathcal{S}_a(\psi_t),{i \in \mathcal{D}_{a,jh}(t)}, \\\\
                f_{a,jh}(t),  \hspace{0.1 cm} \textrm{otherwise,}\\
                \end{array}
              \right. \label{defset}
  \end{equation}

\begin{equation}
  s_{a,h}(t+1)  =\left\{
                \begin{array}{ll}
               0 \hspace{0.2 cm}  \textrm{if} \hspace{0.2 cm} a_t=a_{L,lh}, \hspace{0.1 cm} s_{a,h}(t)=l,  \hspace{0.1 cm} b_t \neq b_{a,lh}  ,\\\\
               l   \hspace{0.2 cm}   \textrm{if} \hspace{0.1 cm} b_t=b_{a,lh}, a_t \in \mathcal{S}_a(\psi_t) ,\\\\
                s_{a,h}(t), \hspace{0.2 cm} \textrm{otherwise.}\\
                \end{array}
              \right. \label{defset}
  \end{equation}
\normalsize

Similarly, the evolution of the defender's state $\psi_{d,t}$ is 

\small
\begin{equation}
 \hspace{-1 cm} \mathcal{D}_{d}(t+1)  =\left\{
                \begin{array}{ll}
                \mathcal{D}_{d}(t)\setminus\{i\} \hspace{0.2 cm}  \textrm{if} \hspace{0.2 cm} a_{t+1}=a_{d,i} \hspace{0.1 cm} b_t \neq b_{d,\tau h}  ,\\
                \mathcal{D}_{d}(t), \hspace{1 cm} \textrm{otherwise,}\\
                \end{array}
              \right. \label{defset}
  \end{equation}

\begin{equation}
  \mathcal{L}_{d,h}(t+1)  =\left\{
                \begin{array}{ll}
                \mathcal{L}_{d,h}(t)\setminus\{l\} \hspace{0.2 cm}  \textrm{if} \hspace{0.2 cm} a_{t+1}=a_{L,lh}, \hspace{0.1 cm} b_t \neq b_{L, h}, \\
                \mathcal{L}_{d,h}(t), \hspace{1 cm} \textrm{otherwise,}\\
                \end{array}
              \right. \label{defset}
  \end{equation}

\begin{equation}
  \mathcal{D}_{d,jh}(t+1)  =\left\{
                \begin{array}{ll}
                \mathcal{D}_{d,jh}(t)\setminus\{i\} \hspace{0.2 cm}  \textrm{if} \hspace{0.2 cm} a_{t+1}=a_{d,i} \hspace{0.1 cm} i \in \mathcal{D}_{d,jh}(t)\\
                \hspace{2.6 cm} b_t \neq b_{d,\tau h},\\
                \mathcal{D}_{d,jh}(t), \hspace{1 cm} \textrm{otherwise,}\\
                \end{array}
              \right. \label{defset}
  \end{equation}

\begin{equation}
  f_{d,jh}(t+1)  =\left\{
                \begin{array}{ll}
                0 \hspace{0.7 cm}  \textrm{if} \hspace{0.2 cm} a_{t+1}=a_{d,i}, \hspace{0.1 cm}  i \in \mathcal{D}_{a,jh}(t+1),\\
                \hspace{1.3 cm}f_{a,jh}(t+1)=i, b_t \in \mathcal{S}_{d,t}(\psi_t,a_t), \\\\
                 i \hspace{0.7 cm} \textrm{if}\hspace{0.1 cm} b_t= b_{c,ijh}, \hspace{0.1 cm} a_{t+1} \neq a_{d,i},  \\\\
                f_{d,jh}(t), \hspace{0.1 cm} \textrm{otherwise,}\\
                \end{array}
              \right. \label{defset}
  \end{equation}

\begin{equation}
  s_{d,h}(t+1)  =\left\{
                \begin{array}{ll}
               0 \hspace{0.2 cm}  \textrm{if} \hspace{0.2 cm} a_{t+1}=a_{L,lh}, \hspace{0.1 cm} s_{d,h}(t)=l,  \hspace{0.1 cm} b_t \in \mathcal{S}_{d,t}(\psi_t,a_t)  ,\\\\
               l   \hspace{0.2 cm}   \textrm{if} \hspace{0.1 cm} b_t=b_{a,lh}, a_{t+1} \neq a_{L,lh},\\\\
                s_{d,h}(t), \hspace{0.2 cm} \textrm{otherwise.}\\
                \end{array}
              \right. \label{defset}
  \end{equation}

\normalsize

The attacker's payoff at each stage $t$ is expressed in terms of its utility which is the sum $S_{D,t}(a_t,b_t,\psi_t)$ of weights of all nodes that will be disconnected from the GS, and the cost $C_{a,t}(a_t,b_t)$ of destroying node $i$, as follows:
\begin{equation}
P_{a,t}(a_t,b_t,\psi_t)=S_{D,t}(a_t,b_t,\psi_t)- \nu C_{a,t}(a_t,b_t,\psi_t),
\end{equation}
where $\nu$ is a normalization constant. The defender's payoff at stage $t$ is expressed in terms of its utility minus its costs. The costs include the sum of weights of disconnected nodes $S_{D,t}(a_t,b_t,\psi_t)$, the transmission time $\Lambda_t(a_t,b_t,\psi_t)$ required to deliver the information to the GS, and the cost of deploying a new node $C_{d,t}(a_t,b_t,\psi_t)$, as follows:
\begin{eqnarray}
\hspace{-0.2 cm}P_{d,t}(a_t,b_t,\psi_t)&=& U_{d,t}(a_t,b_t,\psi_t)-\eta S_{D,t}(a_t,b_t,\psi_t) \nonumber\\
&&\hspace{-0.5 cm}-\mu \Lambda_t(a_t,b_t,\psi_t) - \lambda C_{d,t}(a_t,b_t,\psi_t),\label{dpayoff}
\end{eqnarray}
where $\eta$, $\mu$ and $\lambda$ are normalization constants. For readability, the expressions of $S_{D,t}(a_t,b_t)$,  $\Lambda_t(a_t,b_t,\psi_t)$, $C_{a,t}(a_t,b_t,\psi_t)$, $C_{d,t}(a_t,b_t,\psi_t)$ and $U_d(a_t,b_t,\psi_t)$ in terms of each pair of the attacker's and defender's pure strategies $a_t$ and $b_t$ are given in Appendix B.

%Further, it is shown [] that the mixed strategy Stackelberg equilbrium of the leader is at least as good as the Nash equilibrium of the simultaneous game.
To increase the uncertainty of its action and improve its payoff, the attacker will use a mixed strategy $\boldsymbol{q}_t$ at each stage $t$, thus randomizing its choices across its pure strategies.  The defender, on the other hand, responds with a pure strategy $b_t$ \cite{SSG}. It is assumed that the defender can perfectly observe the strategy of the attacker at each stage $t$.
%In order to increase the uncertainity of the attacker as well as the defender, mixed strategies will be considered for both the attacker and the defender where $b_t$ and $\boldsymbol{q}_t$  are  the vectors of mixed strategies of the defender and the attacker respectively at stage $t$.
 The objective of the attacker is then to find the optimal mixed strategies $\boldsymbol{q}_1, \boldsymbol{q}_2,...,\boldsymbol{q}_T$ that maximize the sum of its expected payoffs until stage $T$ 
\vspace{-0.2 cm}
\begin{eqnarray}
\hspace{-0.3 cm}\max_{\boldsymbol{q}_1, \boldsymbol{q}_2,...,\boldsymbol{q}_T} \sum_{t=1}^T \sum_{a_t \in \mathcal{S}_{a,t}} q_{a_t}P_{a,t}(a_t,b_t,\psi_t) \hspace{0.2 cm} \textrm{s.t.} \hspace{0.1 cm} \boldsymbol{1}\cdot \boldsymbol{q}_t=1 \hspace{0.3 cm} \forall  t,\label{attackopt}
%\vspace{-0.5 cm}
\end{eqnarray}
where $q_{a_t}$ is the probabilitiy with which the attacker chooses action $a_t$. Similarly, the objective of the defender is to find the optimal strategies $b_1, b_2,...,b_T$ that maximizes the sum of its expected payoffs up to stage $T$ i.e.
\begin{equation}
\hspace{-0.3 cm} \max_{b_1, b_2,...,b_T} \sum_{t=1}^T \sum_{a_t \in \mathcal{S}_{a,t}}q_{a_t}P_{d,t}(a_t,b_t,\psi_t) \hspace{0.1 cm} \textrm{s.t.} \hspace{0.1 cm} b_t \in \mathcal{S}_{d,t}(\psi_t,a_t) \hspace{0.1 cm} \forall  t. \hspace{-0.4 cm}\label{defopt}
\end{equation}
Since the attacker's and the defender's actions are coupled to the current stage $t$ and the state $\psi_t$, the \emph{feedback Stackelberg equilibrium} will be used as a solution, as discussed next.
\vspace{-0.2 cm}
\section{Feedback Stackelberg Solution}
The FSE applies for situations in which the leader first chooses its strategy at time instant $t$ $t$, and, then, the follower chooses its strategy based on the current state and the leader's action. In the proposed IoBT connectivity game, the strategy sets $\mathcal{S}_{a,t}(\psi_t)$ and  $\mathcal{S}_{d,t}(\psi_t,a_t)$  of both the attacker and the defender depend on the current state $\psi_t$. Further, the defender strategy set $\mathcal{S}_{d,t}(\psi_t,a_t)$ at time instant $t$ is dependent only on the attacker's action $a_t$ at the current time instant $t$ according to (\ref{defset}) \cite{FSE}. Thus, the FSE is a  suitable solution for our IoBT connectivity game.  The FSE is subgame perfect and time consistent. Thus, at each stage $t$ of the game, the FSE considers the immediate payoff at stage $t$ as well as the expected sum of payoffs of the subsequent stages up to $T$, in contrast to static Stackelberg games which only consider the immediate payoff at stage $t$.
Hence, the FSE solution is obtained recursively using dynamic programing and solving a Stackelberg game at each stage $t$ of the game.   Further, the dynamic nature of the FSE solution makes it adaptive to system changes at any instant $t$. In \cite{FSEstochastic}, it is shown that the FSE remains stable under stochastic Markovian perturbations of the system. The robustness and adaptability of the FSE solution is desirable for a dynamic IoBT system that is constantly subject to random changes due to adversarial conditions.

 Let $ \boldsymbol{q}=(\boldsymbol{q}_1,\boldsymbol{q}_2,..., \boldsymbol{q}_T)$ and $\boldsymbol{b}=(b_1,b_2,...,b_T)$ be respectively the strategy vectors of the attacker and the defender respectively. The FSE strategy will be
\begin{definition}
The strategy profile ($ \boldsymbol{q}^*$, $\boldsymbol{b}^*$) constitute a \emph{feedback Stackelberg equilibrium}
if $\forall \psi_t \in \mathcal{X},$ $t \in \mathcal{T}$,
\begin{equation}
\vspace{-0.1 cm}
\Omega_{a,t}(\boldsymbol{q}^*_{t},b^*_t, \psi_t)=\max_{\boldsymbol{q}_t \in \mathcal{M}_{a,t}} \max_{b_t \in \mathcal{R}^d(\boldsymbol{q}_t)}\hspace{-0.2 cm}\Omega_{a,t}(\boldsymbol{q}_t,b_t, \psi_t), 
\vspace{-0.1 cm}
\end{equation}
where  $\mathcal{M}_{a,t}$ is the space of mixed strategies of the attacker at stage $t$, $\Omega_{a,t}(\boldsymbol{q}_t,b_t, \psi_t)$ is the expected payoff of the attacker starting from stage $t$ and for a state $\psi_t$, and $\mathcal{R}^d(\boldsymbol{q}_t)$ is the optimal strategy set of the defender to the mixed strategy $\boldsymbol{q}_t$ of the attacker and is given by
\vspace{-0.1 cm}
\begin{eqnarray}
\mathcal{R}^d(\boldsymbol{q}_t)= \{b'_t \hspace{0.1 cm} s.t. \hspace{0.1 cm} b'_t= \arg\max_{b_t}\Omega_{d,t}(\boldsymbol{q}_t,b_t, \psi_t) \}, \label{followset}
\vspace{-0.1 cm}
\end{eqnarray}
\end{definition}
for every $b_t \in \cap_{a_t \hspace{0.1 cm} s.t.  \hspace{0.1 cm} q_{a_t}>0} S_{d,t}(\psi_t,a_t)$, where  $\Omega_{d,t}(\boldsymbol{q}_t,b_t, \psi_t)$ is the expected payoff of the defender starting from stage $t$.
At an FSE, the expected payoffs at stage $t$ and for state $\psi_t$ are computed recursively as
\small
\begin{eqnarray}
\hspace{-0.2 cm}\Omega_{a,t}(\boldsymbol{q}_t,b_t, \psi_t)&=&\sum_{a_t \in \mathcal{S}_{a,t}}q_{a_t}\Omega_{a,t+1}(\boldsymbol{q}^*_{t+1},b^*_{t+1}, \psi_{t+1}(a_t,b_t))\nonumber\\
&&+\sum_{a_t \in \mathcal{S}_{a,t}}q_{a_t}P_{a,t}(a_t,b_t), \label{aexpectedpayoff}
\end{eqnarray}
\vspace{-0.4 cm}
\begin{eqnarray}
\hspace{-0.2 cm}\Omega_{d,t}(\boldsymbol{q}_t,b_t, \psi_t)&=&\sum_{a_t \in \mathcal{S}_{a,t}}q_{a_t}\Omega_{d,t+1}(\boldsymbol{q}^*_{t+1},b^*_{t+1}, \psi_{t+1}(a_t,b_t))\nonumber\\
&&+\sum_{a_t \in \mathcal{S}_{a,t}}q_{a_t}P_{d,t}(a_t,b_t),  \label{followrecur}
\end{eqnarray}
\vspace{-0.4 cm}
\\ \normalsize
with $\Omega_{a,T+1}=0$,  $\Omega_{d,T+1}=0$.
%\begin{prop}
%\emph{In order to find the optimal strategy $\boldsymbol{q}^*_t$, the attacker consider the 
%\begin{equation}
%\end{equation}
%\end{prop}
From  (\ref{followset}) and (\ref{followrecur}), we can directly find the optimal action of the defender for a given attacker action at stage $t$ as follows.
\begin{remark}
Given an attacker strategy profile $\boldsymbol{q}_t$, the defender chooses the action $b'_t$ such that
\begin{eqnarray}
&&\hspace{-1 cm} \sum_{a_t \in \mathcal{S}_{a,t}}\hspace{-0.2 cm}q_{a_t}(P_{d,t}(a_t,b_t)+\Omega_{d,t+1}(\boldsymbol{q}^*_{t+1},b^*_{t+1}, \psi_{t+1}(a_t,b_t)) ) \nonumber\\
&&\hspace{-1 cm}\leq \hspace{-0.2 cm} \sum_{a_t \in \mathcal{S}_{a,t}}\hspace{-0.3 cm}q_{a_t}(P_{d,t}(a_t,b'_t)+\Omega_{d,t+1}(\boldsymbol{q}^*_{t+1},b^*_{t+1}, \psi_{t+1}(a_t,b'_t))), \nonumber\\
&&\hspace{-0.5 cm}\forall b_t \in \cap_{a_t \hspace{0.1 cm} \textrm{s.t.}  \hspace{0.1 cm} q_{a_t}>0} S_{d,t}(\psi_t,a_t). \label{defbestresponse}
\end{eqnarray}
\end{remark}
%\begin{proof}
%The result follows directly from (\ref{followset}) and (\ref{followrecur}).
%\end{proof}

In order to find the optimal  mixed strategy of the attacker (i.e. the leader) at each stage $t$, the leader usually solves a linear program for each particular strategy $b'_t$ chosen by the follower (as in \cite{mixedstack}). Then, it chooses, as optimal solution, the mixed strategy of the optimization problem that has the highest payoff. The proposed solution in \cite{mixedstack} is considered when the leader and follower's actions are not coupled. However, in our problem and as shown in (\ref{defset}), the follower's strategy set is coupled to the network state $\psi_t$ and to the attacker's action $a_t$. In other words, if the follower chooses action $b'_t$ as its optimal action and when the network state $\psi_t$, it means that the attacker has chosen its action from the subset $\mathcal{S'}_{a,t}(b'_t, \psi_t)$ of the strategy set $\mathcal{S}_{a,t}$. We define $\mathcal{J}_{\psi_{a,t}}=\{(j,h) \hspace{0.1 cm} \textrm{s.t.} \hspace{0.1 cm} N_{a,jh}(t) = N_{\textrm{th},jh}\}$. Then, the set  $\mathcal{S'}_{a,t}(b'_t, \psi_t)$ is obtained as 
\begin{equation}
   \mathcal{S'}_{a,t}(b'_t, \psi_t)= 
\begin{cases}
   \mathcal{S}_{a,t}\setminus \mathcal{R}_{a,t}, & \text{if } \mathcal{J}_{\psi_{a,t}}  \neq \phi,  b'_t \in \mathcal{V}_{t},\\
    \mathcal{S}_{a,t},              & \text{otherwise},
\end{cases} \label{attackset}
\end{equation}
%\[
%   \mathcal{S'}_{a,t}(b'_t, \psi_t)= 
%\begin{cases}
%   \mathcal{S}_{a,t}\setminus \mathcal{R}_{a,t}, & \text{if } \mathcal{J}_{\psi_{a,t}}  \neq \phi,  b'_t \in \mathcal{V}_{j,t}\\
%    \mathcal{S}_{a,t},              & \text{otherwise}
%\end{cases}
%\]
%\vspace{-0.5 cm}
%\begin{equation}
%\label{attackset}
%\end{equation}
%we modify in our problem the proposed linear program in \cite{mixedstack} to take into account the coupled strategy sets. 
where $\mathcal{R}_{a,t}=\{a_{d,i} \hspace{0.1 cm} \textrm{s.t.} \hspace{0.1 cm} i \in \mathcal{D}_{jh}(t), (j,h) \in \mathcal{J}_{\psi_{a,t}}\}$, $\mathcal{V}_{t}=\mathcal{Q}_d \setminus \{b_{d,ih},  \forall i \hspace{0.1 cm} \textrm{s.t} \hspace{0.1 cm} j \in \mathcal{H}_i \forall j \in \mathcal{Y}_{ih}(t)\}$, the set 
$\mathcal{Y}_{ih}(t)=\{j \in \mathcal{I} \hspace{0.1 cm} \mid \hspace{0.1 cm}i \in \mathcal{D}_{a,jh}(t), 
N_{d,jh}(t)<N_{\textrm{th},jh}\}$, and the set $\mathcal{Q}_d$ is defined in (\ref{kd}). Thus, at each stage $t$, the attacker solves the following linear program for each strategy $b'_t$ of the defender and given a network state $\psi_t$
\small
\begin{eqnarray}
&&\hspace{-0.5 cm}\max_{q_{a,t}}\hspace{-0.2 cm} \sum_{a_t \in \mathcal{S'}_{a,t}(b'_t, \psi_t)}\hspace{-0.8 cm}q_{a_t}(P_{a,t}(a_t,b'_t)+\Omega_{a,t+1}(\boldsymbol{q}^*_{t+1},b^*_{t+1}, \psi_{t+1}(a_t,b'_t)) ),\nonumber\\
&&\hspace{-0.5 cm}\text{s.t.}  \sum_{a_t \in \mathcal{S'}_{a,t}(b'_t, \psi_t)}\hspace{-0.8 cm}q_{a_t}=1, \nonumber\\
&& \hspace{-0.8 cm}\forall b_t, \sum_{a_t \in \mathcal{S'}_{a,t}(b_t, \psi_t)}\hspace{-0.8 cm}q_{a_t}(P_{d,t}(a_t,b_t)+\Omega_{d,t+1}(\boldsymbol{q}^*_{t+1},b^*_{t+1}, \psi_{t+1}(a_t,b_t)) ) \nonumber\\
&&\hspace{-0.2 cm}\leq \hspace{-0.5 cm} \sum_{a_t \in \mathcal{S'}_{a,t}(b'_t, \psi_t)}\hspace{-0.8 cm}q_{a_t}(P_{d,t}(a_t,b'_t)+\Omega_{d,t+1}(\boldsymbol{q}^*_{t+1},b^*_{t+1}, \psi_{t+1}(a_t,b'_t)) ).
\label{mixedStacklinear}
\end{eqnarray}
\normalsize
\vspace{-0.1 cm}

In the proposed linear program, for each strategy $b'_t$ of the defender and for a given state $\psi_t$, the attacker determines the optimal probability for each one of its actions according to (\ref{mixedStacklinear}) while taking into account the best response of the defender defined in (\ref{defbestresponse}).
The proposed linear program holds under the assumption that the attacker has full knowedge of the defender's actions and payoffs at each time epoch $t$. The conditions under which  the IoBT network remains connected are derived in the following section.
\vspace{-0.3 cm}
\normalsize
\section{Connectivity Conditions}

In the IoBT, maintaining connectivity at any time is critical for the successful operation of the network. Thus, the safety objective of connectivity games is more suitable to the IoBT than the reachability objective. In our game, the connectivity of the IoBT network is maintained at each stage $t$ if the attacker does not choose an action that causes disconnection. Disconnection occurs if neither a cluster nor an entire subarea gets disconnected from the GS.  To determine the connectivity conditions under which the IoBT network remains connected, when the FSE solution is used, we first determine, for each action $b'_t$ of the defender, the set $\mathcal{Z}_D(b'_t)$ of attacker's actions that cause disconnection:
%we first determine the connectivity conditions on the attackers actions given that the optimal action of the defender is $b'_t$. To achieve this end,

\small
 \begin{equation}
   \mathcal{Z}_D(b'_t)= \mathcal{S'}_{a,t}(b'_t, \psi_t) \cap
\begin{cases}
 \mathcal{V}_1, & \text{if } b'_t=b_{c,ijh},\\
\mathcal{V}_2, , & \text{if } b'_t=b_{d, \tau h},\\
\mathcal{V}_3, & \text{if } b'_t=b_{a, lh},\\
\mathcal{V}_4, & \text{if } b'_t=b_{L, h}.\\
\end{cases}
\end{equation}
\normalsize
where the set $\mathcal{V}_1=\{a_{f_{j'h'}(t)}, j' \neq j \hspace{0.1 cm} \text{or}  \hspace{0.1 cm} h' \neq h\}\cup \{ a_{s_h(t)}, \forall h\}$, $\mathcal{V}_2=\{a_{f_{jh}(t)},  \forall j,h\}\cup \{ a_{s_h(t)}\forall h\}$, $\mathcal{V}_3=\{a_{f_{jh}(t)}, \forall j,h\}\cup \{ a_{s_h'(t)}, h' \neq h\}$ and $\mathcal{V}_4=\{a_{f_{jh}(t)}, \forall j,h\}\cup \{ a_{s_h(t)}, \forall h\}$.

Thus, given that the optimal action of the defender is $b'_t$, disconnection does not occur if the attacker does not choose an action from the set $\mathcal{Z}_D(b'_t)$.  Let $F_{d,t}(a_t,b_t)=P_{d,t}(a_t,b_t)+\Omega_{d,t+1}(\boldsymbol{q}^*_{t+1},b^*_{t+1}, \psi_{t+1}(a_t,b_t)) )$ and $F_{a,t}(a_t,b_t)=P_{a,t}(a_t,b_t)+\Omega_{a,t+1}(\boldsymbol{q}^*_{t+1},b^*_{t+1}, \psi_{t+1}(a_t,b_t)) )$. The following proposition provides sufficient conditions for the IoBT network to remain connected when the FSE solution is used.

\begin{proposition}
\emph{The proposed FSE solution($\boldsymbol{q}^*, \boldsymbol{b}^*$) maintains connectivity of the IoBT network if for every $\psi_t \in \mathcal{X}$, $t \in \mathcal{T}$.
For each attacker's action $ a^d_t$ in $\mathcal{Z}_D(b^*_t)$, there exists  $a^n_t$ in $\mathcal{S'}_{a,t}(b^*_t, \psi_t) \setminus \mathcal{Z}_D(b^*_t)$ in which one of the following conditions hold:
\begin{enumerate}
\small
\item $\mathcal{B}_{3,t}(b^*_t) = \phi$.
\\
\item If  $\mathcal{B}_{1,t}(b^*_t) \cup \mathcal{B}_{2,t} (b^*_t) \neq \phi$, $W \cdot F_{a,t}(a^n_t,b^*_t) >  F_{a,t}(a^d_t,b^*_t)$.\\
\item If  $\mathcal{B}_{1,t}(b^*_t) \neq \phi$, $\mathcal{B}_{2,t}(b^*_t) \neq \phi$, and  $\min_{b_t \in \mathcal{B}_{1,t}(b^*_t)} \frac{F_{d,t}(a^d_t,b_t)-F_{d,t}(a^d_t,b^*_t)}{F_{d,t}(a^n_t,b_t)-F_{d,t}(a^n_t,b^*_t)}<1$, \\
$\arg \min_{b_t \in \mathcal{B}_{1,t}(b^*_t)} \frac{F_{d,t}(a^d_t,b_t)-F_{d,t}(a^d_t,b^*_t)}{F_{d,t}(a^n_t,b_t)-F_{d,t}(a^n_t,b^*_t)} $
\\
 \hspace{1 cm} $\leq \arg \max_{b_t \in \mathcal{B}_{2,t}(b^*_t)} \frac{F_{d,t}(a^d_t,b_t)-F_{d,t}(a^d_t,b^*_t)}{F_{d,t}(a^n_t,b_t)-F_{d,t}(a^n_t,b^*_t)}$.
\\
\item If $\mathcal{B}_{1,t}(b^*_t) \neq \phi$, $\mathcal{B}_{2,t}(b^*_t) \neq \phi$, and  $\min_{b_t \in \mathcal{B}_{1,t}(b^*_t)} \frac{F_{d,t}(a^d_t,b_t)-F_{d,t}(a^d_t,b^*_t)}{F_{d,t}(a^n_t,b_t)-F_{d,t}(a^n_t,b^*_t)} \geq 1$, \\
$\arg \max_{b_t \in \mathcal{B}_{2,t}(b^*_t)} \frac{F_{d,t}(a^d_t,b_t)-F_{d,t}(a^d_t,b^*_t)}{F_{d,t}(a^n_t,b_t)-F_{d,t}(a^n_t,b^*_t)}\geq 1$.
\\
\item If  $\mathcal{B}_{1,t}(b^*_t) \cup \mathcal{B}_{2,t}(b^*_t)  = \phi$,  $F_{a,t}(a^n_t,b^*_t)>0$.
\end{enumerate}
where\\
\small
$\mathcal{B}_{1,t}(b^*_t)=\{ b_t \in \mathcal{S}_{dn,t}  \mid   F_{d,t}(a^n_t,b_t)-F_{d,t}(a^n_t,b^*_t)\geq 0, F_{d,t}(a^d_t,b_t)-F_{d,t}(a^d_t,b^*_t)\geq 0\}$,\\
$\mathcal{B}_{2,t}(b^*_t)=\{b_t \in \mathcal{S}_{dn,t}\hspace{0.1 cm }\mid \hspace{0.1 cm} F_{d,t}(a^n_t,b_t)-F_{d,t}(a^n_t,b^*_t)< 0, F_{d,t}(a^d_t,b_t)-F_{d,t}(a^d_t,b^*_t)\leq 0\}$,\\
$\mathcal{B}_{3,t}(b^*_t)=\{b_t \in \mathcal{S}_{dn,t}  \hspace{0.1 cm } \mid  \hspace{0.1 cm } F_{d,t}(a^n_t,b_t)-F_{d,t}(a^n_t,b^*_t)\geq 0, F_{d,t}(a^d_t,b_t)-F_{d,t}(a^d_t,b^*_t)< 0\}$,\\
\normalsize
$\mathcal{S}_{dn,t}=\mathcal{S}_{d,t}(\psi_t,a^d_t) \cap {S}_{d,t}(\psi_t,a^n_t)$,
\begin{equation}
W=\frac{F_{d,t}(a^d_t,b^m_t)-F_{d,t}(a^d_t,b^*_t)}{F_{d,t}(a^n_t,b^m_t)-F_{d,t}(a^n_t,b^*_t)},\nonumber
\end{equation}
 \[b^m_t=
 \begin{cases}
 \arg \min_{b_t \in \mathcal{B}_{1,t}(b^*_t)} \frac{F_{d,t}(a^d_t,b_t)-F_{d,t}(a^d_t,b^*_t)}{F_{d,t}(a^n_t,b_t)-F_{d,t}(a^n_t,b^*_t)},& \text{if} \hspace{0.1 cm} C_1(b^*_t),\\
\\
 \arg \max_{b_t \in \mathcal{B}_{2,t}(b^*_t)} \frac{F_{d,t}(a^d_t,b_t)-F_{d,t}(a^d_t,b^*_t)}{F_{d,t}(a^n_t,b_t)-F_{d,t}(a^n_t,b^*_t)},& \text{if} \hspace{0.1 cm} C_2(b^*_t), \
\end{cases}
\]
\normalsize
the condition $ C_1(b^*_t)=\mathcal{B}_{1,t}(b^*_t)  \neq \phi, \mathcal{B}_{1,t}(b^*_t) = \phi$, the condition $C_2(b^*_t)=  \mathcal{B}_{2,t}(b^*_t) \neq \phi$,
and the set $\mathcal{Q}_d$ is defined in (\ref{kd}).
%\begin{equation}
%\end{equation}
%\[
%  \begin{cases}
%  F_{a,t}(a^n_t,b'_t)\frac{F_{d,t}(a^d_t,b^m_t)-F_{d,t}(a^d_t,b'_t)}{F_{d,t}(a^n_t,b^m_t)-F_{d,t}(a^n_t,b'_t)} \leq F_{a,t}(a^d_t,b'_t) & \textrm{\emph{if} } \frac{F_{d,t}(a^d_t,b^m_t)-F_{d,t}(a^d_t,b'_t)}{F_{d,t}(a^n_t,b^m_t)-F_{d,t}(a^n_t,b'_t)} \leq 1\\\\
% F_{a,t}(a^n_t,b'_t) \leq F_{a,t}(a^d_t,b'_t) & \textrm{\emph{otherwise}}
%\end{cases}
%\]
%$\forall a^d_t \in \mathcal{Z}_D(b^*_t) $
}
\label{propconn}
\end{proposition}
%The conditions corresponds to the winning conditions of the defender for the safety objective.

\begin{proof}
Since the attacker uses mixed strategies in our problem,  disconnection does not occur at stage $t$ if $q_{a_t}=0$  for every $a_t$ in  $\mathcal{Z}_D(b'_t)$, i.e. the actions in $\mathcal{Z}_D(b^*_t)$ are dominated.  In  \cite[Corollary 4]{zerolinear}, the conditions are derived for the case in which a variable has a zero value in any optimal solution for a given linear program. In particular, given a linear program of the form $\max \sum_{j=1}^n c_j x_j \hspace{0.1 cm} \text{s.t.}  \hspace{0.1 cm} \sum_{j=1}^n a_{ij}x_j \leq b_j$, $x_j \geq 0$, the variable $x_r=0$ in the optimal solution of the linear program if there exists $q\neq r$ such that one of the following conditions hold:
\begin{enumerate}
\item $I_3 \neq \phi$,
\item If $I_1 \cup I_2 \neq \phi$ then $H c_q \leq c_r$
\item If $I_1 \neq \phi$ and $I_2 \neq \phi$ then $\min_{i \in I_1} \lfloor\frac{a_{ir}}{a_{iq}}\rfloor \geq \max_{i \in I_2}\frac{a_{ir}}{a_{iq}}$
\item if $I_1 \cup I_2 = \phi$ then $c_q>0$,
\end{enumerate}
where $I_1=\{i|a_{iq}>0  \hspace{0.1 cm} \textrm{and}  \hspace{0.1 cm} a_{ir} \geq 0\}$, $I_2=\{i|a_{iq}<0  \hspace{0.1 cm} \textrm{and}  \hspace{0.1 cm} a_{ir} \leq 0\}$,  $I_3=\{i|a_{iq} \geq 0  \hspace{0.1 cm} \textrm{and}  \hspace{0.1 cm} a_{ir} < 0\},$
\begin{equation}
  k  =\left\{
                \begin{array}{ll}
               \arg \min \frac{a_{ir}}{a_{iq}}  \hspace{0.2 cm}  \textrm{if} \hspace{0.1 cm} I_1 \neq \phi,\\\\
                \arg \max \frac{a_{ir}}{a_{iq}}    \hspace{0.2 cm}   \textrm{if}  \hspace{0.1 cm}  I_1 = \phi \hspace{0.1 cm}  \textrm{and} \hspace{0.1 cm}  I_2 \neq \phi\\
                \end{array}
              \right. \label{defset} \nonumber
  \end{equation}

\begin{equation}
  H  =\left\{
                \begin{array}{ll}
               \lfloor \frac{a_{kr}}{a_{kq}} \rfloor  \hspace{0.2 cm}  \textrm{if} \hspace{0.1 cm} I_1 \neq \phi,\\\\
                \arg \max \frac{a_{kr}}{a_{kq}}    \hspace{0.2 cm}   \textrm{if}  \hspace{0.1 cm}  I_1 = \phi \hspace{0.1 cm}  \textrm{and} \hspace{0.1 cm}  I_2 \neq \phi\\
                \end{array}
              \right. \label{defset} \nonumber
  \end{equation}

Thus, by applying these conditions to each $q_{a_t}$ in $\mathcal{Z}_D(b^*_t)$ our proposed linear program in (\ref{mixedStacklinear}),  the result follows. \vspace{-0.2 cm}
\end{proof}
%Proposition \ref{propconn} implies that, at FSE, the attacker does not commit a disconnection action $a^d_t$ at stage $t$, if there exists an action $a^n_t$ that does not cause disconnection and there exists no defender's action $b_t$ other than the defender's action $b^*_t$ at FSE such that both of the following conditions hold. The first condition is that the defender's payoff of action $b^*_t$ is less than or equal to the defender's payoff of action $b_t$, when the  attacker's action is $a^n_t$. The second condition is that defender's payoff of action $b^*_t$ is strictly greater than the defender's payoff of action $b_t$, when the  attacker's action is $a^d_t$.

Proposition $\ref{propconn}$ shows that maintaining connectivity at each time epoch $t$ depends on the payoffs of the attacker and the defender. Further, the payoff of the defender in (\ref{dpayoff}) depends on the IoBT network parameters. For example,  the time required to deliver the information to the GS is a function of the IoBT network capacity, as shown in Appendix B, which can be controlled by adjusting the transmission bandwidth in a wireless setting.
Thus, in order to maintain connectivity at each time epoch $t$, the IoBT operator adjusts its payoffs such that one of the conditions in  Proposition $\ref{propconn}$ is met.

%when choosing action $a^n_t$ and the defender's action at FSE $b^*_t$  is positive and either one of the following conditions hold
%%1)The defender's payoff of attacker's action $a^d_t$ and the defender's action at FSE $b^*t$ is greater than the  defender's payoff of attacker's action $a^d_t$ and any other defender's action $b_t \neq b^*t$

\section{Simulation Results and Analysis}
For our simulations, we consider an IoBT network containing $1000$ devices of seven types: Type 1 corresponds to a radiological sensor, type $2$ corresponds to a chemical sensor, type $3$ corresponds to an infrared (IR) camera, type $4$ corresponds to an explosives detector, type 5 corresponds to a surveillance camera, type 6 corresponds to a mititary robot containing a chemical sensor, a radiological sensor, an infrared camera and an explosives detector, and type $7$ corresponds to a military unmanned vehicle containing a surveillance camera, an IR camera, a radiological sensor and a chemical sensor. The number of subareas considered is $H=5$. The number of LSs available in each subarea is $L_h=2$, the weight of each LS $i$ is set to $w_{L,i}=15$, which is chosen to be greater than the weight of any of the devices at a lower hierarchy level. The threshold on the number of sensors in each cluster $\mathcal{D}_{jh}$ is set to $N_{\textrm{th},jh}=15$. The normalizing coefficients are set to: $\mu=100$, $\nu=1$, and $\lambda=1$. The costs of deploying a device of type $\tau$ and an LS are set to be $d_\tau=0.5N_\tau$ and $d_L=50$. All normalization constants and cost values are chosen such that the costs are comparable to the number of sensors.
%The utility of deploying a device of type $\tau$ and an LS are set to be $u_\tau=N_\tau$ and $u_L=20$.
For detailed analysis, the following scenarios are considered:
\begin{enumerate}
\item The cost of compromising an LS $c_L$ is varied  between 0 and 200 in steps of 50. The considered value of the cost $c_{aL}$ of determining the activated LS by the attacker is set to $0$. The maximum number of stages considered is $T=1,2,3$, where $T=1$ corresponds to the case of Stackelberg equilibrium with no feedback (NFSE). The cost of determining the CH is set to be $c_{CH}=20$ while the cost of compromising a device of type $\tau$ is set to be  $c_{\tau}=0.5 N_\tau$.
%For the considered set of costs values, the expected number of disconnected sensors and the probability of attacking the activated LS in the subrea with the highest weight are computed when FSE is used and for the attacker's policy of compromising any of the activated LSs with equal probability.
\item The cost of finding the CH is varied between $0$ and $100$ in steps of $20$. The costs of compromising an LS and determining the activated LS $(c_{aL},c_L)$ is set to $(100,50)$. The maximum number of stages considered are $T=1,2,3$.
%For the considered set of costs values, the expected number of disconnected sensors and the probability of attacking the CH of the cluster with highest weight are computed when FSE is used and for the attacker's policy of compromising any of the activated LSs with equal probability.
\item The maximum number of stages $T$ is varied between $1$ and $5$ in steps of $1$. The considered cost values are $c_{\tau}=0.5 N_\tau$ and $c_{CH}=20$, and the LS costs $(c_{aL},c_L)$ are to set to $(0,50)$ and $(150,50)$, respectively. 
%The expected number of disconnected sensors is computed for each value of the maximum number of stages when FSE and when Stackelberg equilibrium with no feedback are used respectively.
\end{enumerate}

%\begin{figure}[t]
%\centering     %%% not \center
%\subfigure[Probability of attacking the LS with the highest weight.]{\label{LSprobability}\includegraphics[width=60mm]{LSprobability10.pdf}}
%\subfigure [ The expected number of disconnected nodes.]{\label{LScost}\includegraphics[width=60mm]{LScost10.pdf}}
%\caption{Probability of attacking the LS with the highest weight and the expected number of disconnected nodes respectively vs the cost of compromising an LS}
%\vspace{-0.3 cm}
%\end{figure}
\begin{figure}[t]
	\centering
	\includegraphics[width=7 cm,height=4cm,angle=0]{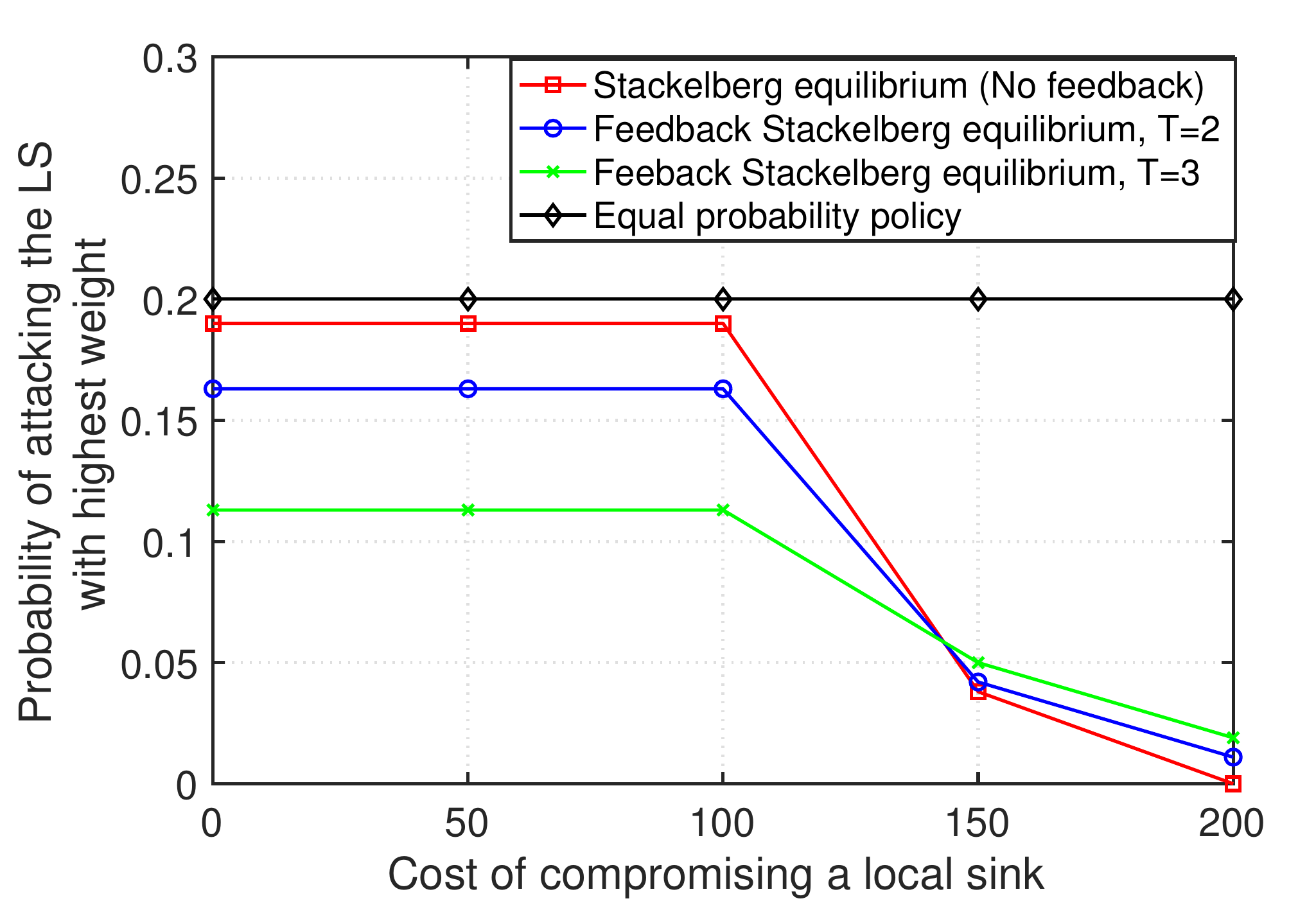}
	\caption{ Probability of attacking the LS with the highest weight vs the cost of compromising an LS
	}\vspace{-0.4 cm}\label{LSprob}
\end{figure}

\begin{figure}[t]
	\centering
	\includegraphics[width=7 cm,height=4cm,angle=0]{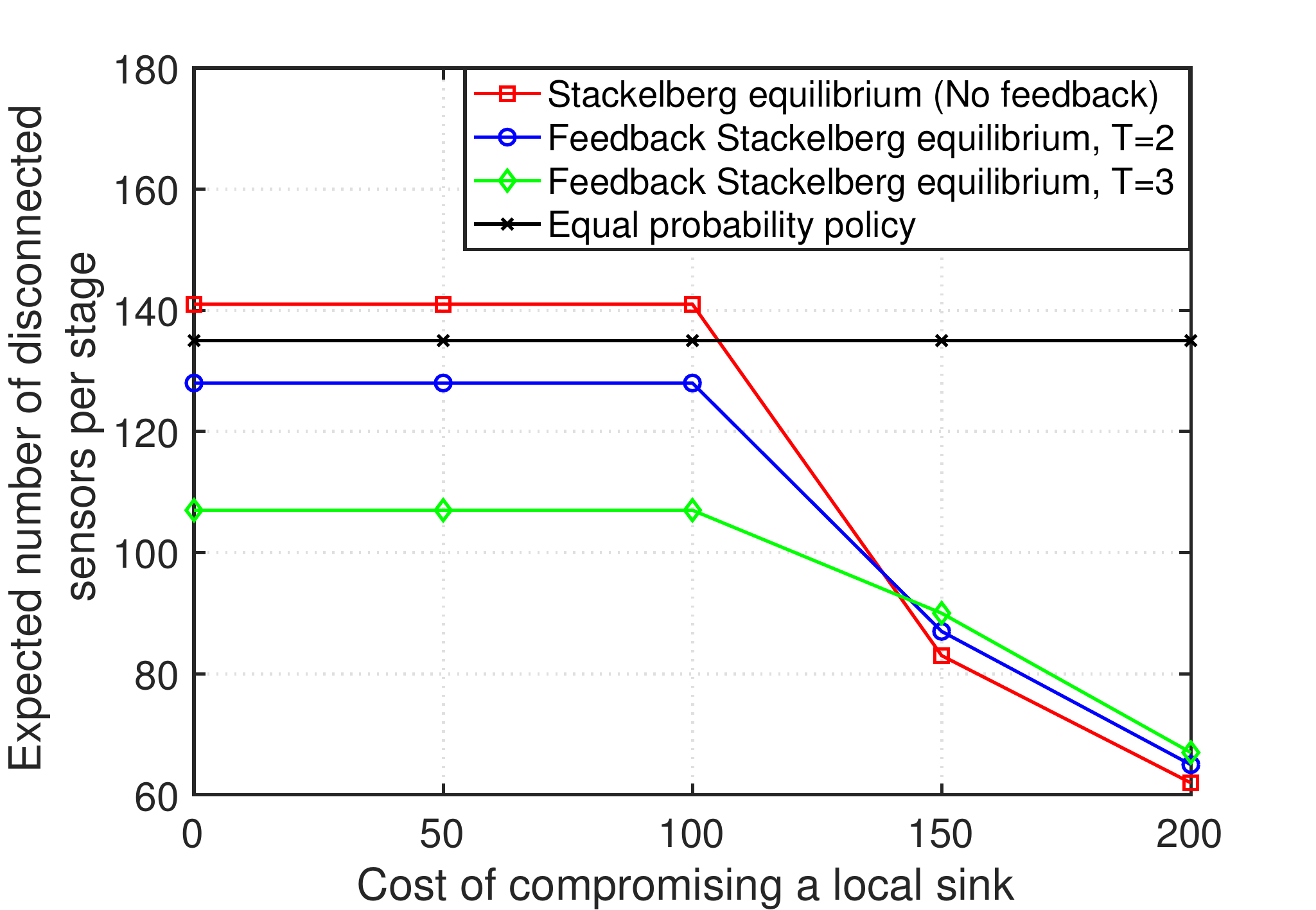}
	\caption{ Expected number of disconnected nodes respectively vs the cost of compromising an LS
	}\vspace{-0.4 cm}\label{LScost}
\end{figure}

Fig.  \ref{LSprob} shows, for the NFSE, the FSE, and a baseline policy which assigns equal probabilities to attacking the activated LSs, the average probability $p_H$ of attacking the LS in the subarea with the highest weight as a function of the cost $c_L$ of compromising an LS.
 %when the cost $c_{aL}$ of finding the activated LS is $0$ or $100$.
 Fig.  \ref{LSprob} first shows that the NFSE probability of attacking the activated LS with the highest weight is 0.19 for $c_L$ values less than $100$. In this case, the attacker's payoff obtained from attacking any activated LS is considerably higher than the attacker's payoff achieved from attacking any other device. Thus, the attacker chooses to compomise only the five activated LSs. As a result, the NFSE mixed strategy of the attacker is comparable to the equal probability policy. As $c_L$ increases to $150$, the payoffs achieved by attacking an activated LS and the CHs become comparable, and the attacker chooses to attack both the LSs and the CHs. Thus, the probability of attacking the LS with highest weight decreases to 0.09. When $c_L$ increases further to $200$, the payoff resulting from attacking an activated LS becomes considerably lower than the payoff achieved by attacking any of the CHs, and, thus, the attacker chooses to attack only the CHs.
Next, when using the FSE with $T=2$, the FSE probability of attacking the activated LS with the highest weight decreases to $0.163$ for $c_L$ values less than $100$. This is due to the fact that, by using a FSE, the attacker, as well as the defender, will take into account the  expected sum of payoffs from $t+1$ to $T$, when computing the FSE probabilities at time $t$ according to (19). Thus, the attacker's payoff received from an attack on any IoBT node increases compared to the NFSE case, which results in increasing the probability of attacking some of the IoBT nodes which are not LSs. Hence, the probability of attacking the LS having the highest weight decreases. Then, the probability of attacking the LS with the highest weight decreases as $c_L$ increases to $200$. However, the probability is higher than the case of the NFSE. This is because the expected sum of payoffs achieved by attacking the LS with highest weight increases compared to the NFSE according to (19).
Finally,  when using  the FSE with $T=3$ and for $c_L$ values less than $100$, the probability $p_H$ decreases to $0.112$ compared to the case with $T=2$. This is because the expected payoff of attacking any IoBT node increases with the number of stages according to (19), which causes $p_H$ to decrease. Then, as $c_L$ increases to $200$, $p_H$ decreases, but it remains higher than the case when $T=2$, since the expected sum of payoffs received from attacking the LS with the highest weight increases with $T$.
%Next, for the case in which the cost of finding the activated LS is increased 100 with $c_L=0$ the probability of attacking the activated LS is 0.19, which corresponds to the case in which the attacker chooses to attack only the LSs. From Fig.  \ref{LSprobability}, we can also see that, for $c_L=50$, the attacker chooses to compromise both the LSs and the CHs, and, thus, the probability drops to 0.09. For higher values of $c_L$, the payoff of attacking a CH becomes considerably higher than the payoff of attacking an LS. Thus, the probability of attacking an activated LS becomes zero.

%\begin{figure}[t]
%	\centering
%	\includegraphics[width=8 cm,height=5 cm,angle=0]{LScost.pdf}
%	\caption{Expected number of disconnected sensors vs the cost of compromising an LS
%	}\vspace{-0.4 cm}\label{LScost}
%\vspace{-0.8 cm}
%\end{figure}
Fig. \ref{LScost} shows the average number $N_D$ of disconnected sensors per stage  resulting from the NFSE, the FSE, and the equal probability policy as function of the cost of compromising an LS. First, using the NFSE and when $c_L \leq 100$, the expected number of disconnected sensors is $141$ since the attacker chooses to compromise only the activated LSs. Also, the expected number of disconnected sensors is slightly higher than when the attacker chooses to attack each of the activated LSs with equal probability, since this policy is not optimal. Fig. \ref{LScost} also shows that, for $c_L=150$, the expected number of disconnected sensors decreases to $83$ since the attacker chooses to compromise either the LSs or the CHs. Thus, the value of the expected number of disconnected sensors drops below the value of the equal probability policy.  When $c_L$ increases to $200$, the expected number of disconnected sensors decreases to $62$ since the attacker will now compromise CHs. Next, using FSE with $T=2$ and for $c_L < 100$, $N_D$ decreases by $9\%$ compared to NFSE, since $p_H$ decreases as shown in Fig. \ref{LSprob}. Then, as $c_L$ increases to $200$, the expected number of disconnected sensors decreases yet becomes higher than the NFSE case since $p_H$ is higher according to Fig.\ref{LSprob}.
Finally, when using the FSE with $T=3$, the expected number $N_D$ is $107$ when $c_L < 100$. This is because, for the considered cost values, the probability $p_H$ decreases with $T$ as shown in Fig. \ref{LSprob}. Then, as $c_L$ increases to $200$, the expected number of disconnected sensors decreases, but it remains greater than the case in which $T=2$. This is due to the fact that the probability $p_H$ increases with $T$ when $c_L$ is greater than $150$ as shown in Fig. \ref{LSprob}.
%\begin{figure}[t]
%\centering     %%% not \center
%\subfigure[Probability of attacking the CH with the highest weight.]{\label{CHprob}\includegraphics[width=60 mm]{CHprob10.pdf}}
%\subfigure [ The expected number of disconnected nodes.]{\label{CHcost}\includegraphics[width=60 mm]{CHcost10.pdf}}
%\caption{Probability of attacking the CH with the highest weight and the expected number of disconnected nodes respectively vs the cost of finding a CH.}
%\vspace{-0.5 cm}
%\end{figure}
\begin{figure}[t]
	\centering
	\includegraphics[width=7 cm,height=4 cm,angle=0]{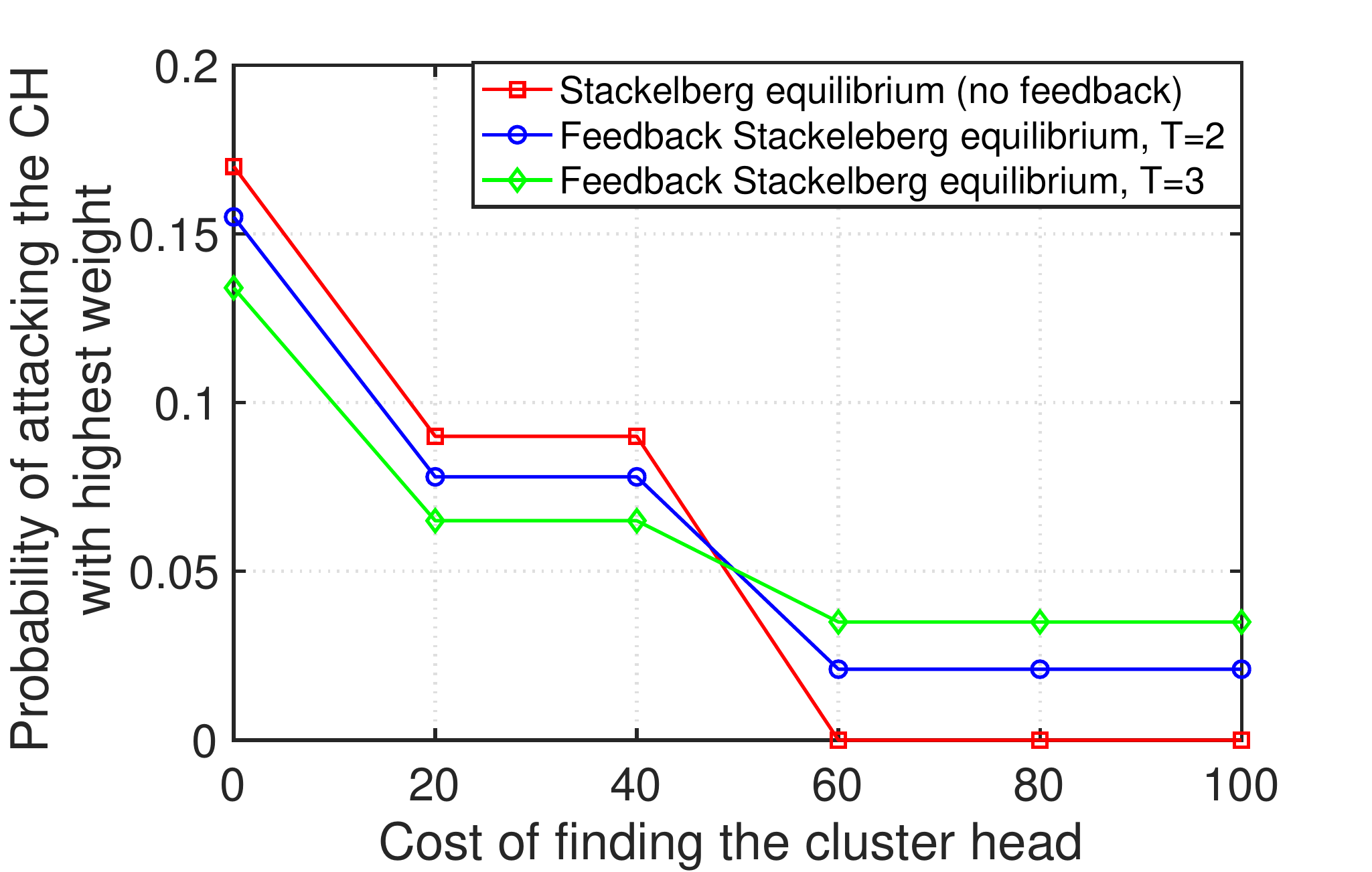}
	\caption{Probability of attacking the CH with the highest weight vs the cost of compromising an LS
	}\vspace{-0.4 cm}\label{CHprob}
\end{figure}

\begin{figure}[t]
	\centering
	\includegraphics[width=7 cm,height=4 cm,angle=0]{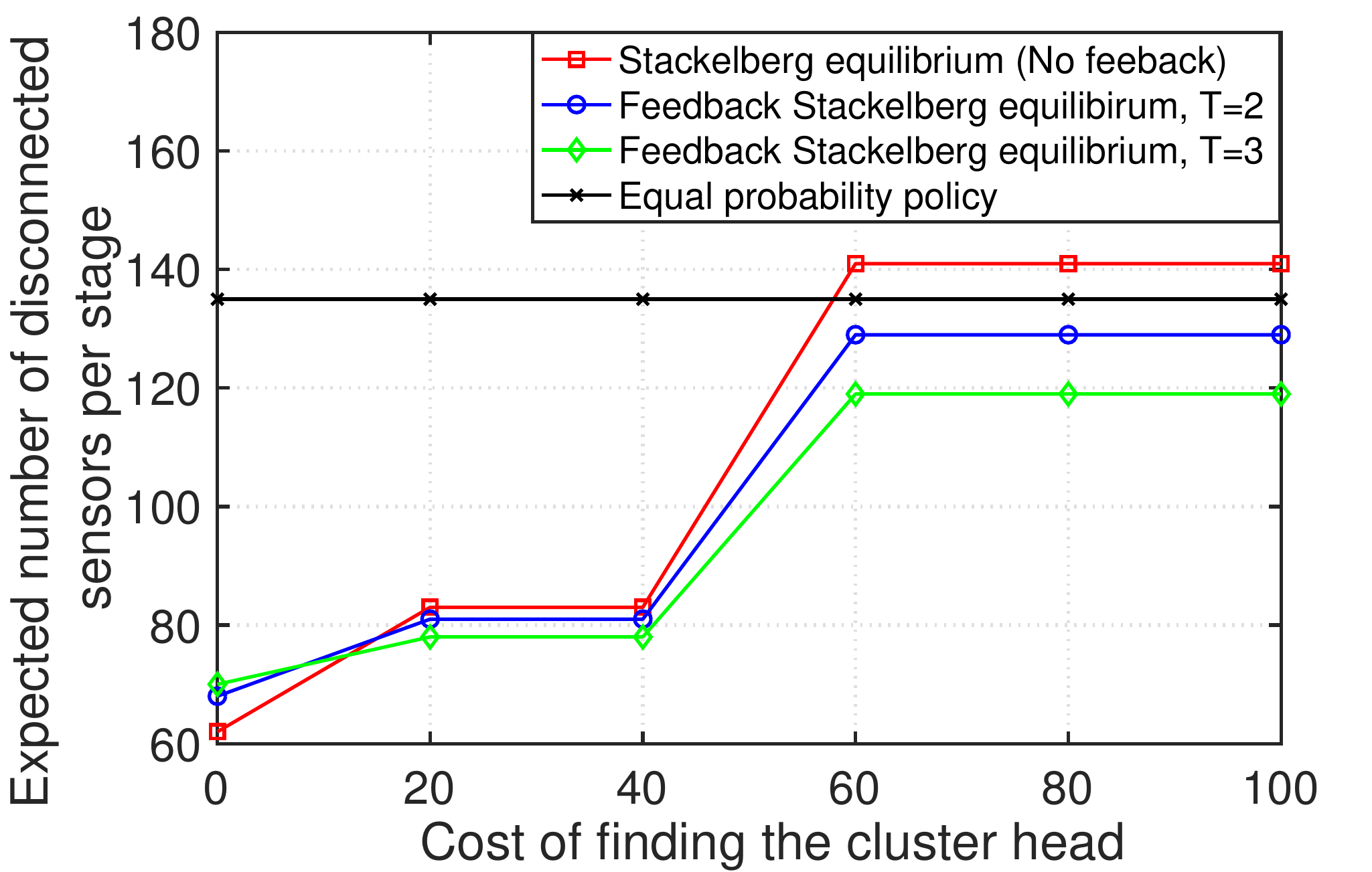}
	\caption{Expected number of disconnected sensors per stage vs the cost of finding the CH
	}\vspace{-0.4 cm}\label{CHcost}
\end{figure}

Fig \ref{CHprob} shows the average probability $p_{c,\max}$ of attacking the CH of the cluster with highest number of sensors resulting from the NFSE and the FSE versus the cost of finding the CH. Using NFSE  and  when $c_{CH}$ is $0$, $p_{c,\max}$ is 0.17 since the payoffs resulting from an attack on the CHs are the highest.  As the value of $c_{CH}$  increases to $40$, $p_{c,\max}$ decreases to 0.09 since the attacker chooses to compromise both LSs and CHs. Then,  $p_{c,\max}$ becomes zero as $c_{CH}$ increases up to 100 since the payoffs achieved by compromising the CHs will be considerably lower than the payoffs obtained from compromising the LSs. Hence, in this case, the attacker will be compromising the LSs.
Next, using the FSE with $T=2$ and when $c_{CH}$ is $0$,  $p_{c,\max}$  is 0.155. Then, when $c_{CH}$ increases up to $40$,  $p_{c,\max}$ decreases to $0.078$. Thus, when $c_{CH}$ is less than 40, the value of $p_{c,\max}$ is
less than its value when NFSE is used.  This is because the expected sum of payoffs achieved by attacking any IoBT node is higher than the NFSE according to (19), which yields a decrease in the probability $p_{c,\max}$. As $c_L$ increases from $40$ to $100$, $p_{c,\max}$ decreases yet its value becomes higher than the NFSE. This is because the expected sum of payoffs received from attacking a CH, when using the FSE, is higher than the case in which the NFSE is used, which causes the probability to remain positive.
Finally, when using the FSE with  $T=3$, $p_{c,\max}$ varies similar to the case in which $T=2$. However, for $c_{CH}$  less than $40$, $p_{c,\max}$ is smaller than the case when $T=2$ since the expected sum of payoffs achieved by attacking any IoBT node increases with $T$. Also, for $c_{CH}$ greater than $40$, $p_{c,\max}$ is higher than the case when $T=2$  since the expected sum of payoffs obtained from attacking a CH increases with $T$ according to (19).  Finally, Fig. \ref{CHprob}  shows that, for the equal probability policy, the probability of attacking a CH is zero since the attacker only compromises the activated LSs.

Fig. \ref{CHcost} shows, for both the FSE and the equal probability policy, the average number of disconnected sensors $N_D$ per stage versus the cost of finding the CH. When using the NFSE and when $c_{CH}=0$, the expected number of disconnected sensors is $63$ since the attacker will be  compromising the CHs. Then, the expected number of disconnected sensors is $82$ as the value of $c_{CH}$ increases to $40$, since the attacker  will choose proper (non-deterministic) mixed strategies over both the LSs and CHs. Then, as $c_{CH}$ becomes higher than $40$, the expected number of disconnected sensors increases to $141$ and exceeds the value of the equal probability policy since the attacker will be compromising the activated LSs and the equal probability policy is not optimal. 
Next, when using the FSE with $T=2$ and when $c_{CH}=0$, the expected number of disconnected sensors $N_D$ is $65$. In this case, the value of $N_D$ is slightly higher than the one resulting from the NFSE since the probability of attacking an LS is positive with FSE. As  $c_{CH}$ increases to $40$, the value of $N_D$ increases to $81$. Then, as $c_{CH}$ increases to $100$, the value of $N_D$ increases to $129$. For $c_{CH}$ values higher than $40$, the value of $N_D$ resulting from FSE is lower than the NFSE case since the probability of attacking a device  which is not a CH is positive with FSE. When using the FSE with $T=3$, the value of $N_D$ varies as function of $c_{CH}$ in a similar way as when $T=2$. Yet, when $c_{CH}$ is $0$, the value of $N_D$ is slightly higher than when $T=2$. This is because in this case the probability of attacking an LS increases with $T$. Also, for $c_{CH}$ values greater than $40$, the value of $N_D$ is lower than when $T=2$ since the probability of attacking a device which is not a CH increases with $T$ for the considered cost values.

\begin{figure}[t]
	\centering
	\includegraphics[width=7 cm,height=4cm,angle=0]{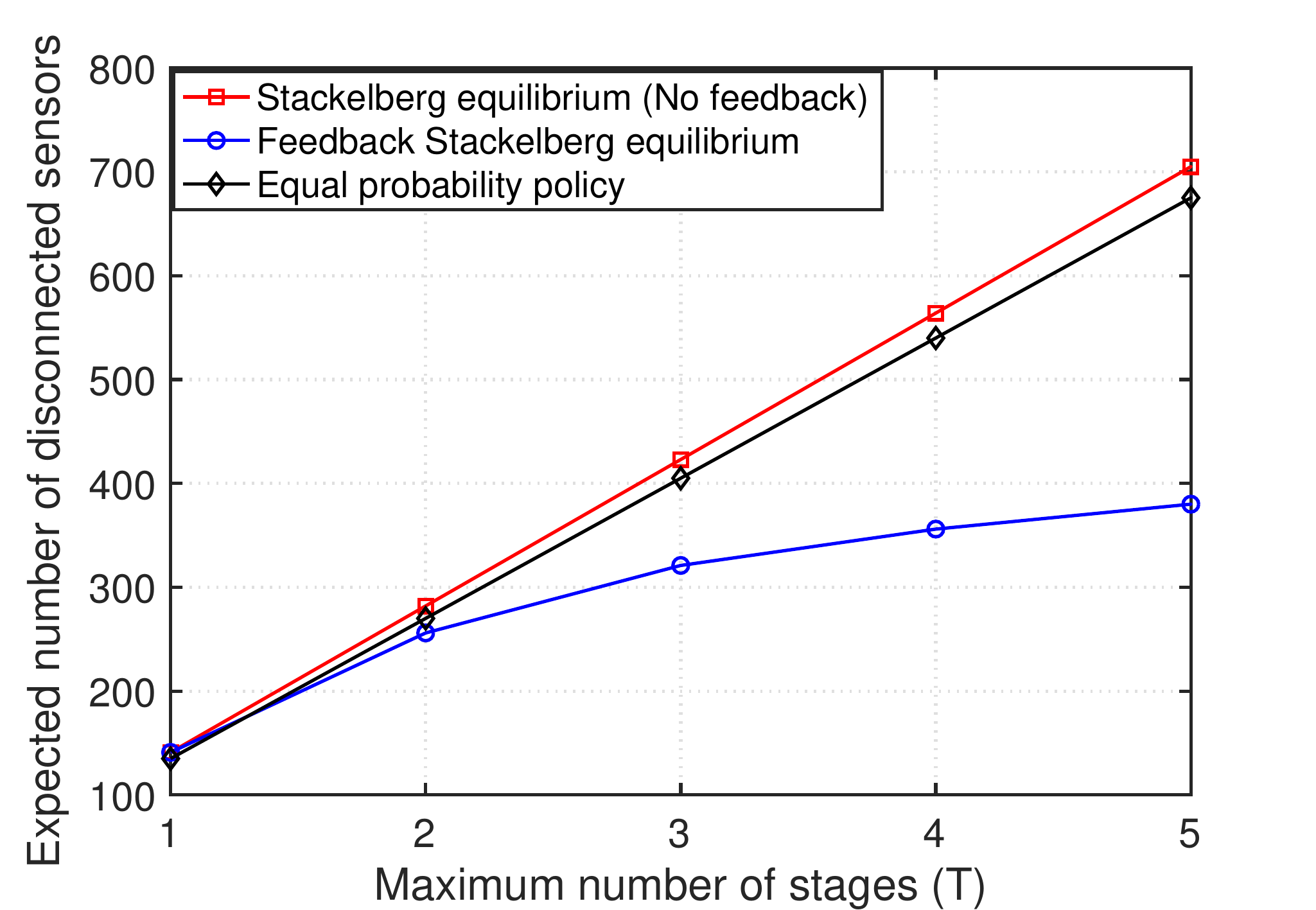}
	\caption{ Expected number of disconnected nodes respectively vs number of stages
	}\vspace{-0.4 cm}\label{largeT}
\end{figure}

\begin{figure}[t]
	\centering
	\includegraphics[width=7 cm,height=4cm,angle=0]{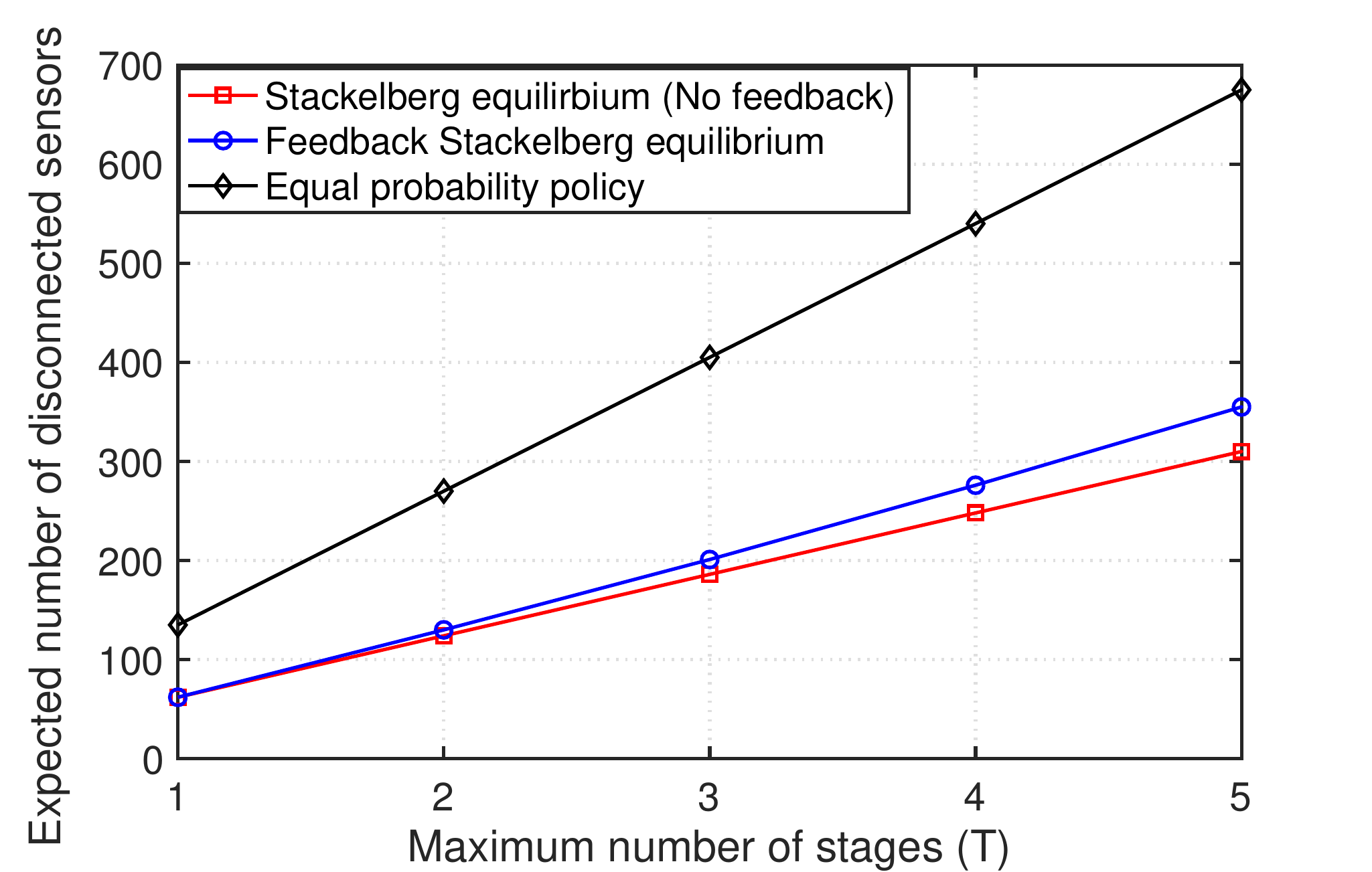}
	\caption{ Expected number of disconnected nodes respectively vs number of stages
	}\vspace{-0.4 cm}\label{smallT}
\end{figure}

Figs. \ref{largeT} shows, for the case when the LSs costs  $(c_{aL},c_L)$ are  $(0,50)$, the expected number of disconnected sensors versus the maximum number of stages $T$ when FSE, the NFSE, and the equal probabilitiy policy are used, respectively. The NFSE solution corresponds to finding the Stackelberg equilibrium for a one stage game played at each time epoch $t$ ($1 \leq t \leq T$).  In Fig. \ref{largeT},   the expected number of disconnected sensors increases with the maximum number of stages using the three solutions. However, the expected number of disconnected sensors, when the FSE is used, increases at a rate considerably slower than when either the NFSE or the equal probability power policy are used. Thus, the results confirm that, using the FSE, the number of disconnected sensors per stage decreases with $T$ as opposed to the  NFSE in which the attacker chooses  its mixed strategy only over the LSs.
%This is because for the considered cost values, the attacker will choose its mixed strategy only over the LSs using the Stackelberg equilibrium solution with no feedback. In the FSE, however, the attacker uses its mixed strategy over the LSs as well as other IoBT nodes since the payoffs of the IoBT nodes increase with feedback according to (\ref{aexpectedpayoff}), which causes 
The decrease in the number of disconnected sensors when using the FSE reaches up to $43\%$ compared to the equal probability policy and up to  $46\%$ compared to the Stackelberg equilibrium with no feedback, when  $T$ is $5$. 

Fig. \ref{smallT} shows, for the case when the LSs costs  $(c_{aL},c_L)$ are  $(150,50)$, the expected number of disconnected sensors versus the maximum number of stages $T$ when FSE, the NFSE solution, and the equal probabilitiy policy are used, respectively.
In this case,  the expected number of disconnected sensors also increases with the maximum number of stages using the three solutions. However, the number of disconnected sensors using the FSE is slightly higher than the case of NFSE. This is due to the fact that, in the case of no feedback, the attacker's mixed strategy does not include attacking the LSs  whereas, using the FSE, the attacker's mixed strategy includes attacking the LSs, and the probability of attacking the LSs increases with the number of stages. The increase in the number of disconnected sensors reaches up to $14\%$ when $T$ is $5$. Nonetheless, the  number of disconnected sensors, using the FSE, decreases by $47\%$ compared to the equal probability policy.

Thus, Figs. \ref{largeT} and \ref{smallT} show that the FSE yields a significant decrease in the number of disconnected sensors when the LS costs are low. The performance is slightly degraded compared to the NFSE when the costs of attacking the LSs increase. However, the number of disconnected sensors remains significantly lower than the equal probability policy.

\section{Conclusion}

In this paper, we have considered the connectivity problem in an Internet of Battlefield Things network in which an adversary attempts to cause disconnection by compromising one of the IoBT nodes at each time epoch while a defender tries to restore the connectivity of the IoBT by deploying new IoBT nodes or changing the roles of nodes. We have formulated the problem as a multistage Stackelberg game in which the attacker is the leader and the defender is the follower. Due to the reliance of the attacker's and the defender's actions on the network state at each stage $t$, we have adopted the feedback Stackelberg equilibrium to solve the game.
We have obtained sufficient condition to maintain connectivity at each stage $t$ when the FSE solution is used.
 Numerical results show that the expected number of disconnected sensors, when the FSE solution is used, decreases up to $46\%$ compared to a baseline scenario in which a Stackelberg game with no feedback is used,
 and up to $43\%$ compared to a baseline equal probability policy.

%\begin{subappendices}
%\renewcommand{\thesection}{\Alph{section}}%
%% or try \arabic{section}
%
%\section{Also you should know this}
%Really.
%\section{And I also came across this}
%But I need to put this in an appendix so that my paper is not too long.
%\end{subappendices}

%The expected number of disconnected sensors is computed
%\renewcommand{\thesection}{\Alph{section}}%
% or try \arabic{section}

%\section{Also you should know this}

%\begin{appendices}

\def\baselinestretch{0.95}

\appendix
\section*{Appendix A: Summary of Notation}

\begin{description}

\mysymbol{\zzz}{\mathcal{D}}{Set of IoBT devices.}
\mysymbol{\Txxx}{\mathcal{K}}{Set of devices' types.}
\mysymbol{\vvvv}{\mathcal{K}}{Size of set $\mathcal{K}$}
\mysymbol{\vvvv}{\mathcal{I}}{Set of information types.}
\mysymbol{\vvvv}{\mathcal{M}}{Size of set $\mathcal{I}$}
\mysymbol{\vvvvv}{N_\tau}{Number of devices of type $\tau$.}
\mysymbol{\vvvvv}{\mathcal{H}_\tau}{Information set of a device of type $\tau$.}
\mysymbol{\vvvvv}{A_h}{IoBT subarea $h$.}
\mysymbol{\vvvvv}{\mathcal{D}_{jh}}{Cluster sensing information $j$ in subarea $A_h$.}
\mysymbol{\vvvvv}{N_{th,jh}}{Minimum required number of sensors in cluster $\mathcal{D}_{jh}$.}
\mysymbol{\vvvvv}{w_i}{Weight of device $i$.}
\mysymbol{\vvvvv}{w_{L,i}}{Weight of LS $i$.}
\mysymbol{\vvvvv}{\mathcal{L}_h}{Set of LSs in subarea $A_h$.}
\mysymbol{\vvvvv}{c_\tau}{Cost of compromising a device of type $\tau$.}
\mysymbol{\vvvvv}{c_L}{Cost of compromising an LS.}
\mysymbol{\vvvvv}{c_{CH}}{Cost of determining the cluster head.}
\mysymbol{\vvvvv}{c_{aL}}{Cost of determining the activated LS.}
\mysymbol{\vvvvv}{d_\tau}{Cost of deploying a device of type $\tau$.}
\mysymbol{\vvvvv}{d_L}{Cost of deploying an LS.}
\mysymbol{\vvvvv}{\mathcal{P}}{Set of players.}
\mysymbol{\vvvvv}{\mathcal{T}}{Set of stages.}
\mysymbol{\vvvvv}{\mathcal{X}}{State space.}
\mysymbol{\vvvvv}{\mathcal{S}_{a,t}}{Attacker's strategy set at time $t$.}
\mysymbol{\vvvvv}{\mathcal{S}_{d,t}}{Defender's strategy set at time $t$.}
\mysymbol{\vvvvv}{\psi_t}{State of the game at time $t$.}
\mysymbol{\vvvvv}{\psi_{a,t}}{Network observed by the attacker at time $t$.}
\mysymbol{\vvvvv}{\mathcal{D}_{a}(t)}{\hspace{0.1 cm}Set of devices in $\psi_{a,t}$.}
\mysymbol{\vvvvv}{\mathcal{D}_{a,jh}(t)}{\hspace{0.2  cm} Cluster $\mathcal{D}_{jh}$  in $\psi_{a,t}$.}
\mysymbol{\vvvvv}{f_{a,jh}(t)}{\hspace{0.1 cm} Index of the device that is CH of cluster $\mathcal{D}_{a,jh}(t)$.}
\mysymbol{\vvvvv}{ \mathcal{L}_{a,h}(t)}{\hspace{0.1 cm} Set of LSs in subarea $A_h$ at time $t$.}
\mysymbol{\vvvvv}{s_{a,h}(t)}{\hspace{0.1 cm} Index of the activated LS in $\mathcal{L}_{a,h}(t)$.}
\mysymbol{\vvvvv}{\psi_{d,t}}{Network observed by the defender at time $t$.}
\mysymbol{\vvvvv}{\mathcal{D}_{d}(t)}{\hspace{0.1 cm}Set of devices in $\psi_{d,t}$.}
\mysymbol{\vvvvv}{\mathcal{D}_{d,jh}(t)}{\hspace{0.3 cm}Cluster $\mathcal{D}_{jh}$ in $\psi_{d,t}$.}
\mysymbol{\vvvvv}{f_{d,jh}(t)}{\hspace{0.1 cm} Index of the device that is CH of cluster $\mathcal{D}_{d,jh}(t)$.}
\mysymbol{\vvvvv}{ \mathcal{L}_{d,h}(t)}{\hspace{0.1 cm} Set of LSs in subarea $A_h$ in $\psi_{d,t}$.}
\mysymbol{\vvvvv}{s_{a,h}(t)}{\hspace{0.1 cm} Index of the activated LS in $\mathcal{L}_{a,h}(t)$.}
\mysymbol{\vvvvv}{a_t}{Attacker's action at time $t$.}
\mysymbol{\vvvvv}{b_t}{Defender's action at time $t$.}
\mysymbol{\vvvvv}{a_{d,i}}{Attacking device $i$.}
\mysymbol{\vvvvv}{a_{L,lh}}{Attacking LS $l$ in subarea $A_h$.}
\mysymbol{\vvvvv}{b_{c,ijh}}{Assigning device $i$ to be the CH of $\mathcal{D}_{jh}$.}
\mysymbol{\vvvvv}{b_{d,\tau h}}{Deploying a device of type $\tau$ in subarea $A_h$.}
\mysymbol{\vvvvv}{b_{a,l h}}{Activating LS $l$ in subarea $A_h$.}
\mysymbol{\vvvvv}{\boldsymbol{q}_t}{Attacker's mixed strategy at time $t$.}
\mysymbol{\vvvvv}{P_{a,t}}{Attacker's payoff at time $t$.}
\mysymbol{\vvvvv}{P_{d,t}}{Defender's payoff at time $t$.}
\mysymbol{\vvvvv}{S_{D,t}}{Sum of weights of disconnected sensors at time $t$.}
\mysymbol{\vvvvv}{\Lambda_{t}}{Delay to deliver the information at time $t$.}

\end{description}
%\section{zzz}
%
%\[ \zzz=\Txxx\]
%where
%\[\vvvv=3\]
%and
%\[\vvvvv=X\]

%\vspace{-0.2 cm}
\section*{Appendix B: Expressions of Payoff Functions}
The expressions of $S_{D,t}(a_t,b_t)$,  $\Lambda_t(a_t,b_t)$ in terms of each pair of the attacker's and defender's pure strategies $a_t$ and $b_t$ are given as follows.
\vspace{0.3 cm}
\begin{itemize}
\item If $a_t=a_{d,i}$, $b_t=b_{c,kj'h'},$ 
\smaller
\begin{eqnarray}
&&\hspace{-1 cm}S_{D,t}(a_t, b_t)=\sum_{h=1}^H\sum_{j=1}^M I(i \in \mathcal{D}_{jh}(t))\big(x_{ijh}(t)W_{jh}(t)+\bar{x}_{ijh}(t)) \big)\nonumber\\
&&\hspace{0.7 cm}-z_{j'h'}(t)W_{j'h'}(t),\nonumber
\end{eqnarray}
\vspace{-0.5 cm}
%\begin{eqnarray}
%&&\hspace{-1 cm}\Lambda_t(a_t,b_t)=\sum_{h=1}^H\sum_{j \in \mathcal{L}_h}I(i \in \mathcal{D}_{jh}(t)) \bar{x}_{ijh}(t)(\Lambda_{jh}(f_{jh}(t),s_h(t),D_{jh}(t) \setminus \{i\})\nonumber\\
%&&\hspace{-0.5 cm}+ I(i \notin \mathcal{D}_{jh}(t), h' \neq h)(\Lambda_{jh'}(f_{jh'}(t),s_h'(t),D_{jh'}(t))\nonumber\\
%&&\hspace{-0.5 cm}+\Lambda_g(s_{h'}(t)))+\Lambda_{jh}(k,s_h(t),\mathcal{D}_{jh}(t))+\Lambda_g(s_h(t))\nonumber
%\end{eqnarray}
\begin{eqnarray}
&&\hspace{-1 cm}\Lambda_t(a_t,b_t)=\max_{1 \leq h \leq H} \bar{z}_h(t) \max_{1 \leq j \leq M}(1-I(h=h',j=j'))\bar{z}_{jh}(t)\nonumber\\
&& \hspace{-0.5 cm}\times (I(i \in \mathcal{D}_{jh}(t))\bar{x}_{ijh}(t)\Lambda_{jh}(f_{jh}(t),s_{h}(t),\mathcal{D}_{jh}(t) \setminus \{i\})\nonumber\\
&& \hspace{-0.5 cm}+I(i \notin \mathcal{D}_{jh}(t))\Lambda_{jh}(f_{jh}(t),s_{h}(t),\mathcal{D}_{jh}(t) ))\nonumber\\
%&& \hspace{-1 cm}+I(i \notin \mathcal{D}_{jh}(t))\bar{z}_{jh}(t) \nonumber
&&\hspace{-0.5 cm}
+ I(h=h',j=j')(I(i \in \mathcal{D}_{jh}(t))
\Lambda_{jh}(k,s_{h}(t),D_{jh}(t) \setminus \{i\})\nonumber\\
&&\hspace{-0.5 cm}+(I(i \notin \mathcal{D}_{j'h}(t))
\Lambda_{jh}(k,s_{h}(t),D_{jh}(t) ))+\Lambda_g(s_h(t)), \nonumber
\end{eqnarray}

\normalsize

where for any variable $x$, $\bar{x}=1-x$, \emph{I(.)} is the indicator function. In this part, all the network variables pertains to network $\psi_{a,t}$, and the index $a$ is dropped for ease of notation.  $x_{ijh}(t)$ is an indicator whether device $i$ is the CH in cluster $\mathcal{D}_{jh}(t)$, $f_{jh}(t)$ is the CH of $\mathcal{D}_{jh}(t)$, $s_h(t)$ is the activated LS in subarea $h$, $W_{jh}(t)$ is given by $W_{jh}(t)=N_{jh}(t)$, $z_h(t)$ is an indicator if subrea $A_h$ is currently disconnected from the GS, $z_{jh}(t)$ is an indicator if cluster $\mathcal{D}_{jh}(t)$ is currently disconnected from the network,  $\Lambda_{jh}(f_{jh}(t), s_h(t), \mathcal{D}_{jh}(t))$ is the time to transmit the information from cluster $\mathcal{N}_{jh}(t)$ to LS $s_h(t)$ and is given by  $\Lambda_{jh}(f_{jh}(t), s_h(t), \mathcal{D}_{jh}(t))=\Lambda(f_{jh}(t), s_h(t))+\max_{n \in \mathcal{D}_{jh}(t)}\Lambda(n,f_{jh}(t))$,  $\Lambda(f_{jh}(t), s_h(t))$ is the time needed to transmit the information from CH $f_{jh}(t)$ to LS $s_h(t)$, $\Lambda(n,f_{jh}(t))$ is the time needed to transmit the information from device $n$ to CH $f_{jh}(t)$, and $\Lambda_g(s_h(t))$ is the time required to deliver the information from $s_h(t)$ to the GS. For any two IoBT nodes $i$ and $j$, $\Lambda(i,j)$ is the single hop delay between $i$ and $j$ and is given by: $\Lambda(i,j)=\frac{m_i}{R_{ij}}$ where $R_{ij}$ is the capacity of the link $(i,j)$ and $m_i$ is the packet size of node $i$.
%\begin{eqnarray}
%%P_{d,t}(a_t, b_t, \psi_t)&=&-x_{ijh}(t)W_{jh}(t)-(1-x_{ijh}(t))w_i - c_{d,i} \nonumber\\
%\end{eqnarray}
%\item If $a_t=a_{d,i}$, $b_t=b_{c,kjh}$, $i \in \mathcal{D}_{jh}(t)$
%\begin{equation}
%S_{D,t}(a_t, b_t, \psi_{t} )=w_i-c_{d,i} \nonumber
%\end{equation}
\vspace{0.2 cm}
\item  If $a_t=a_{d,i}$, $b_t=b_{d,kh'}$, $i \in \mathcal{D}_{h''}(t)$, $h'' \neq h'$,
\vspace{-0.1 cm}
\smaller
\begin{eqnarray}
&&\hspace{-1 cm}S_{D,t}(a_t, b_t)=\sum_{j=1}^M I(i \in \mathcal{D}_{jh''(t)})(x_{ijh''}(t)W_{jh''}(t) +\bar{x}_{ijh''}(t)),  \nonumber
\end{eqnarray}
\vspace{- 0.5  cm}
\begin{eqnarray}
&&\hspace{-1 cm}\Lambda_t(a_t,b_t)=\max_{1 \leq h \leq H} \bar{z}_h(t) I(h \neq h')\max_{1 \leq j \leq M}\bar{z}_{jh}(t)\nonumber\\
&&\hspace{-0.8 cm}\times (I(i \notin \mathcal{D}_{jh}(t)) (\Lambda_{jh}(f_{jh}(t),s_{h}(t),D_{jh}(t))\nonumber\\
&&\hspace{-0.8 cm}+I(i \in \mathcal{D}_{jh}(t)) I(f_{jh}(t) \neq i) (\Lambda_{jh}(f_{jh}(t),s_{h}(t),D_{jh}(t) \setminus \{i\})\nonumber\\
&&\hspace{-0.8 cm}+\Lambda_g(s_{h}(t)))+I(h=h')\max_{1 \leq j \leq M} I(j \notin \mathcal{H}_k)(f_{jh'}(t),s_h'(t),\mathcal{D}_{jh'}(t))\nonumber\\
&&\hspace{-0.8 cm}+ I(j \in \mathcal{H}_k)\Lambda_{jh'}(f_{jh'}(t),s_{h'}(t),\mathcal{D}^+_{jh'}(t))+\Lambda_g(s_{h'}(t))),\nonumber
\end{eqnarray}
\normalsize
where $\mathcal{D}_{h}(t)$ is the set of devices in subarea $A_h$, $\mathcal{D}^+_{jh}(t)=\mathcal{D}_{jh}(t)\cup\{N(t)+1\}$, and $N(t)+1$ is the index of the newly deployed device and $N(t)$ is the total number of devices.
\vspace{0.3 cm}
\item If $a_t=a_{d,i}$, $b_t=b_{d,kh'}$, $i \in \mathcal{D}_{h'}(t)$,
\smaller
\begin{eqnarray}
&&\hspace{-1 cm}S_{D,t}(a_t, b_t)=\sum_{j=1}^M I(i \in \mathcal{D}_{jh'}(t), j \notin \mathcal{H}_k)x_{ijh'}(t)W_{jh'}(t)\nonumber\\
&&  \hspace{0.8 cm}+ \hspace{0.1 cm}I(i \in \mathcal{D}_{jh'}(t),j \in \mathcal{H}_k)x_{ijh'}(t)(W_{jh'}(t)+1)\nonumber\\
&&  \hspace{0.8 cm}+ \hspace{0.1 cm}\bar{x}_{ijh'}(t), \nonumber
\end{eqnarray}
\vspace{-0.5 cm}
\begin{eqnarray}
&&\hspace{-1 cm}\Lambda_t(a_t, b_t)=\max_{1 \leq h \leq H} I(h=h') \bar{z}_h(t)\max_{1\leq j \leq M} I(i \in \mathcal{D}_{jh'}(t)) I(f_{jh'}(t) \neq i)\nonumber\\
&& \hspace{1 cm}\times (I(j \in \mathcal{H}_k)\bar{z}_{jh}(t)\Lambda_{jh'}(f_{jh'}(t),s_{h'}(t),\mathcal{D}^+_{jh}(t) \setminus \{i\})\nonumber\\
&&\hspace{1 cm}+(I(j \notin \mathcal{H}_k)\bar{z}_{jh}(t)\Lambda_{jh'}(f_{jh'}(t),s_{h'}(t),\mathcal{D}_{jh}(t)\setminus \{i\})) \nonumber\\
&&\hspace{1 cm}+I(h \neq h')\max_{1 \leq j \leq M}\bar{z}_{jh}(t) \Lambda_{jh}(f_{jh}(t),s_{h}(t),\mathcal{D}_{jh}(t)) \nonumber\\
&&\hspace{1 cm}+\ \Lambda_g(s_h(t)). \nonumber
\end{eqnarray}
\normalsize
\item If $a=a_{d,i}$, $b_t=b_{a,kh''}$,
%\begin{eqnarray}
%&&S_{D,t}(a_t,b_t)=\sum_{h=1}^H I(i \in \mathcal{D}_h, h \neq h'')\sum_{j=1}^{M}I(i \in \mathcal{D}_{jh})x_{ijh}(t)W_{jh}(t)+\bar{x}_{ijh}(t)w_i \nonumber\\
%&&+I(i \in \mathcal{D}_h, h=h'')(-z_{h''}(t)W_{h''}+\sum_{j =1}^MI(i \in \mathcal{D}_{jh})x_{ijh}(t)W_{jh}(t)+\bar{x}_{ijh}(t)w_i) \nonumber\\
%&&-I(i \notin \mathcal{D}_{h}, h=h'')z_{h''}(t)W_{h''}\nonumber
%\end{eqnarray}
\smaller
\begin{eqnarray}
&&\hspace{-1 cm}S_{D,t}(a_t,b_t)=\sum_{h=1}^H \sum_{j=1}^{M}I(i \in \mathcal{D}_{jh}(t))(x_{ijh}(t)W_{jh}(t)+\bar{x}_{ijh}(t)) \nonumber\\
&&\hspace{1.6 cm}+I( h=h'')(-z_{h''}(t)W_{h''}(t)), \nonumber
\end{eqnarray}
\begin{eqnarray}
&&\hspace{-1 cm}\Lambda_t(a_t, b_t)=\max_{1 \leq h \leq H} (I(h \neq h'')\bar{z}_h(t)+I(h=h''))\nonumber\\
&&\hspace{0.3 cm}\times\max_{1 \leq j \leq M}I(i \in \mathcal{D}_{jh}(t))I(f_{jh}(t)\neq i)\nonumber\\
&&\hspace{0.3 cm}\times\Lambda_{jh}(f_{jh}(t),s_{h}(t),\mathcal{D}_{jh}(t)\setminus \{i\}))\nonumber\\
&&\hspace{0.3 cm}-I(i \notin \mathcal{D}_{jh}(t))\Lambda_{jh}(f_{jh}(t),s_{h}(t),\mathcal{D}_{jh}(t)) + \Lambda_g(s_h(t)), \nonumber
%&&-I(i \in \mathcal{D}_h, h=h'')z_{h''}(t)W_{h''}-\sum_{j \in \mathcal{L}_h}x_{ijh}(t)W_{jh}(t)+\bar{x}_{ijh}(t)w_i \nonumber\\
%&&-I(i \notin \mathcal{D}_{jh}, h=h'')z_{h''}(t)W_{h''}
\end{eqnarray}
\normalsize
where $W_h(t)$ is given by: \small $W_h(t)=\sum_{i \in \mathcal{D}_h(t)}w_i+w_{L,s_h(t)}$.
\normalsize
\item If $a_t=a_{d,i}$, $b_t=b_{L,h''}$,
\smaller
\begin{eqnarray}
&&\hspace{-1 cm}S_{D,t}(a_t, b_t)=\sum_{h=1}^H\sum_{j=1}^M I(i \in \mathcal{D}_{jh}(t))(x_{ijh}(t)W_{jh}(t) +\bar{x}_{ijh}(t)),  \nonumber
\end{eqnarray}
\vspace{-0.3 cm}
\begin{eqnarray}
&&\hspace{-0.5 cm}\Lambda_t(a_t, b_t)=\max_{1 \leq h \leq H}\bar{z}_h(t)\max_{1 \leq j \leq M}I(i \in \mathcal{D}_{jh})I(f_{jh}(t)\neq i)\nonumber\\
&&\hspace{1 cm} \times\Lambda_{jh}(f_{jh}(t),s_{h}(t),\mathcal{D}_{jh}(t)\setminus \{i\}))\nonumber\\
&&\hspace{1 cm}-I(i \notin \mathcal{D}_{jh})\Lambda_{jh}(f_{jh}(t),s_{h}(t),\mathcal{D}_{jh}(t))) + \Lambda_g(s_h(t))).\nonumber
\end{eqnarray}
\normalsize
\item If $a_t=a_{L,mh'}$, $b_t=b_{a,kh''}$, $h'=h''$,\\
\vspace{-0.3 cm}
\smaller
\begin{eqnarray}
&&\hspace{-4.8 cm}S_{D,t}(a_t, b_t)=-z_{h''}(t)W_{h''}(t),\nonumber
\end{eqnarray}
%where $z_h(t)$ is an indicator whether subrea $A_h$ is disconnected from the GS at time instant $t$.
\begin{eqnarray}
&& \hspace{-0.8 cm}\Lambda_t(a_t, b_t)=\max_{1 \leq h \leq H} I(h=h')\hspace{-0.1 cm}\max_{1 \leq j \leq M}\hspace{-0.2 cm}\Lambda_{jh'}(f_{jh'}(t),k,\mathcal{D}_{jh'}(t))\nonumber\\
&& \hspace{0.5 cm}+ I(h \neq h')\bar{z}_{h}(t)(\max_{1 \leq j \leq M}\bar{z}_{jh}(t)\Lambda_{jh}(f_{jh}(t),s_h(t),\mathcal{D}_{jh}(t))\nonumber\\
&& \hspace{0.5 cm}+I(h=h') \Lambda_g(k)+I(h\neq h')\bar{z}_h(t) \Lambda_g(s_h(t)).\nonumber
\end{eqnarray}
\normalsize
\item If $a_t=a_{L,mh'}$, $b_t=b_{a,kh''}$, $h' \neq h''$,
\smaller
\begin{eqnarray}
&&\hspace{-1 cm}S_{D,t}(a_t, b_t)=y_{mh'}(t)W_h(t)+\bar{y}_{mh}(t)w_{L,m}-z_{h''}(t)W_{h''}(t),\nonumber
\end{eqnarray}
\vspace{-0.5 cm}
\begin{eqnarray}
&& \hspace{-1 cm}\Lambda_t(a_t, b_t)=\max_{1 \leq h \leq H}( I(s_{h'}(t)\neq m,h=h') +I(h \neq h' ))\nonumber\\
&& \times \bar{z}_{h}(t)(\Lambda_g(s_h(t))+\max_{1 \leq j \leq M}\bar{z}_{jh}(t)\Lambda_{jh}(f_{jh}(t),s_{h}(t),\mathcal{D}_{jh}(t)))\nonumber\\
&& +I(h=h'')(\max_{1 \leq j \leq M}\bar{z}_{jh''}(t)\Lambda_{jh''}(f_{jh''}(t),k,\mathcal{D}_{jh''}(t))+\Lambda_g(k)), \nonumber
 \end{eqnarray}
\normalsize
where $y_{mh}(t)$ is an indicator that LS $m$ is the activated LS in subarea $A_h$.
\\
\item If $a_t=a_{L,mh'}$, $b_t=b_{L,h''}$,\\
\smaller
\begin{equation}
\hspace{-2.4 cm}S_{D,t}(a_t, b_t)=y_{mh'}(t)W_{h'}(t)+\bar{y}_{mh'}w_{L,m}, \nonumber
\end{equation}
\vspace{-0.5 cm}
\begin{eqnarray}
&&\hspace{-1 cm}\Lambda_t(a_t,b_t)=\max_{1 \leq h \leq H}\bar{z}_h(t)(I(s_{h'}(t) \neq k, h=h')+I( h \neq h'))\nonumber\\
&&\times \max_{1 \leq j \leq M}(\bar{z}_{jh}(t)\Lambda_{jh}(f_{jh}(t),s_{h}(t),D_{jh}) + \Lambda_g(s_h(t)).\nonumber
\end{eqnarray}
\normalsize
\item If $a_t=a_{L,mh'}$, $b_t=b_{d,kh''}$, $h' \neq h''$,
\smaller
\begin{eqnarray}
&&\hspace{-3 cm}S_{D,t}(a_t, b_t)=y_{mh'}(t)W_{h'}(t) +\bar{y}_{mh'}(t)w_{L,m}, \nonumber
\end{eqnarray}
\begin{eqnarray}
&&\hspace{-1 cm}\Lambda_t(a_t,b_t)=\max_{1 \leq h \leq H} (I(s_{h'}(t) \neq m, h=h')+(h \neq h', h \neq h'')) \nonumber\\
&&\hspace{-0.5 cm}\times \bar{z}_h(t)\sum_{j =1}^M\bar{z}_{jh}(t)(\Lambda_{jh}(f_{jh}(t),s_{h}(t),\mathcal{D}_{jh}(t)) \nonumber\\
&&\hspace{-0.5 cm}+I(h=h''))\bar{z}_h(t)(\max_{1 \leq j \leq M}I(j \in \mathcal{H}_k)\bar{z}_{jh}(t)\Lambda_{jh}(f_{jh}(t),s_{h}(t),\mathcal{D}^+_{jh}(t))\nonumber\\
&&\hspace{-0.5 cm}+ I(j \notin \mathcal{H}_k)\bar{z}_{jh}(t)\Lambda_{jh}(_{jh}(t),s_{h}(t),\mathcal{D}_{jh}(t)) + \Lambda_g(s_h(t))). \nonumber
\end{eqnarray}
\normalsize
\item If $a_t=a_{L,mh'}$, $b_t=b_{d,kh''}$, $h' = h''$,
\smaller
\begin{eqnarray}
&&\hspace{-2 cm}S_{D,t}(a_t, b_t)=y_{mh}(t)(W_{h'}(t)+N_k)+\bar{y}_{mh'}(t)w_{L,m}, \nonumber
\end{eqnarray}
\begin{eqnarray}
&&\hspace{-1 cm}\Lambda_t(a_t,b_t)=\max_{1 \leq h \leq H}(I(s_h'(t) \neq m, h=h')\bar{z}_h(t) \nonumber\\
&&\hspace{-0.6 cm}\times (\max_{1 \leq j \leq M}I(j \in \mathcal{H}_k)\bar{z}_{jh}(t)\Lambda_{jh}(f_{jh}(t),s_{h}(t),\mathcal{D}^+_{jh}(t))\nonumber\\
&&\hspace{-0.6 cm}+I(j \notin \mathcal{H}_k)\bar{z}_{jh}(t)\Lambda_{jh}(f_{jh}(t),s_{h}(t),\mathcal{D}_{jh}(t)+\Lambda_g(s_h(t))))\nonumber\\
&&\hspace{-0.6 cm}+I(h\neq h')\bar{z}_h(t)(\max_{1 \leq j \leq M}\bar{z}_{jh}(t)\Lambda_{jh}(f_{jh}(t),s_{h}(t),\mathcal{D}_{jh}(t))+\Lambda_g(s_h(t))).\nonumber
\end{eqnarray}
%\item $a_t=a_{L,jh}$, $b_t=b_{d,kh'}$, $h' \neq h$
%\begin{equation}
%P_{a,t}(a_t, b_t)=y_{jh}(t)W_h(t)+(1-y_{jh}(t))W_{L,j}-c_{L,j} \nonumber
%\end{equation}
\normalsize
\item If $a_t=a_{L,mh'}$, $b_t=b_{c,kj'h''}$, 
\smaller
\begin{equation}
\hspace{-3.2 cm}S_{D,t}(a_t, b_t)=y_{mh'}(t)W_{h'}(t)+\bar{y}_{mh'}(t)w_{L,m}, \nonumber
\end{equation}
\begin{eqnarray}
&&\hspace{-1.3 cm}\Lambda_t(a_t, b_t)=\max_{1 \leq h \leq H}(I(s_h(t) \neq m, h=h')+I(h \neq h')) \nonumber\\
&&\times \bar{z}_h(t)(\max_{1 \leq j \leq M}I(h=h'',j=j') \Lambda_{jh}(k,s_h(t),\mathcal{D}_{jh}(t))\nonumber\\
&&+(1-I(h=h'',j=j'))\bar{z}_{jh}(t)\Lambda_{jh}(f_{jh}(t),s_{h}(t),\mathcal{D}_{jh}(t))\nonumber\\
&&+\Lambda_g(s_h(t))).\nonumber
\end{eqnarray}
\normalsize
The expressions of $C_{a,t}(a_t,b_t)$, $C_{d,t}(a_t,b_t)$ and  $U_{d,t}(a_t,b_t)$  in terms of the pure strategies of the attacker and the defender are given as follows:
\smaller
\[
  \hspace{-2 cm}  C_{a,t}(a_t,b_t)= 
\begin{cases}
    c_i,& \text{if } a_t=a_{d,i},\\ 
    c_{L,k},              & \text{if } a_t=a_{L,kh}.
\end{cases}
\]
\[
  \hspace{-2 cm}  C_{d,t}(a_t,b_t)= 
\begin{cases}
    d_i,& \text{if } b_t=b_{d,ih},\\
    d_L,              & \text{if } b_t=b_{L,h},\\
0, &\text{otherwise}.
\end{cases}
\]
\normalsize
The expression of the defender's utility is given by
\smaller
\[
  \hspace{-2 cm}  U_{d,t}(a_t,b_t)= 
\begin{cases}
    u_{i},& \text{if } b_t=b_{d,ih},\\
    u_L,              & \text{if } b_t=b_{L,h},\\
0, &\text{otherwise}.
\end{cases}
\]
\normalsize

%\item If $a_t=a_{d,i}$, $b_t=b_{d,kh}$, $x_{\psi_t}=0$
%\begin{eqnarray}
%P_{a,t}(a_t, b_t, \psi_t)&=&\sum_{j=1}^M I(i \in \mathcal{D}_{jh}(t))(x_{ijh}(t)(W_{jh}(t)+1)\nonumber\\
%&&+(1-x_{ijh}(t))) - c_{d,i} \nonumber
%\end{eqnarray}
\end{itemize}
\end{document}